\documentclass[10pt,a4paper]{article}
\usepackage{todonotes,cancel}
\usepackage[normalem]{ulem}
\usepackage[latin,english]{babel}
\usepackage[utf8]{inputenc}
\usepackage[T1]{fontenc}
\usepackage{eufrak}
\usepackage{geometry}
\usepackage{amsmath}
\usepackage{bbm}
\usepackage{amssymb}
\usepackage{amsbsy}
\usepackage{amsthm}
\usepackage{makeidx}
\usepackage{mathtools}
\usepackage{mathabx, mathrsfs, dsfont}
\usepackage{cases}
\usepackage{braket}
\usepackage[toc,page]{appendix}
\usepackage{dirtytalk}
\usepackage{pdfsync}
\usepackage{hyperref}
\usepackage{graphicx}
\usepackage{epstopdf}
\usepackage{todonotes}

\usepackage{graphicx}

\usepackage{fancyhdr}

\usepackage{math}
\usepackage{cite}
\usepackage{color}
\usepackage{subcaption}
\mathtoolsset{showonlyrefs}

\newcommand{\balpha}{{\boldsymbol{\alpha}}}
\newcommand{\bbeta}{{\boldsymbol{\beta}}}
\newcommand{\btheta}{\boldsymbol{\theta}}
\newcommand{\bgamma}{{\boldsymbol{\gamma}}}

\renewcommand{\T}{\mathbb{T}}
\renewcommand{\meanval}[1]{\bE\left[#1\right]}

\renewcommand{\mod}{{\,{\rm mod}\,}}

\newcommand{\bzeta}{\boldsymbol{\zeta}}

\newtheorem{lemma}{Lemma}[section]
\newtheorem{theorem}[lemma]{Theorem}
\newtheorem{proposition}[lemma]{Proposition}

\newtheorem{remark}[lemma]{Remark}

\newtheorem{definition}[lemma]{Definition}

\numberwithin{equation}{section}

\title{Generalized Gibbs ensemble of  the Ablowitz-Ladik lattice, Circular $\beta$-ensemble  and double confluent  Heun equation}
\author{
	T. Grava
	\footnote{
		International School for Advanced Studies (SISSA), Via Bonomea 265,   34136, Trieste, Italy, INFN sezione di Trieste and School of Mathematics,  University of Bristol,  Fry Building,  BS8 1UG,   UK\newline
		\textit{Email: } \texttt{grava@sissa.it} 
	}  ,
	G. Mazzuca
    \footnote{Department of Mathematics, The Royal Institute of Technology, Lindstedts\"agen 25, 114 28, Stockholm, Sweden, International School for Advanced Studies (SISSA), Via Bonomea 265,   34136, Trieste, Italy \newline
		\textit{Email:} \texttt{mazzuca@kth.se}}
    }

\date{\today}

\begin{document}
	\maketitle
	\begin{abstract}
	We consider the discrete    defocusing nonlinear Schr\"odinger equation in its integrable version, which is called defocusing Ablowitz-Ladik lattice.
		We   consider   periodic boundary conditions with period $N$ and initial data sample according to the Generalized Gibbs ensemble. 
		In this setting, the  Lax matrix  of the Ablowitz-Ladik  lattice is a   random CMV-periodic matrix  and it  is related to the  Killip-Nenciu Circular $\beta$-ensemble at high-temperature.
		We  obtain the  generalized free energy of the Ablowitz-Ladik lattice  and  the  density of states of the random Lax matrix  by  establishing  a mapping to the
		one-dimensional log-gas. For the Gibbs measure related to the Hamiltonian of the Ablowitz-Ladik flow, we obtain the density of states via a particular solution of the double-confluent Heun equation.
				 	\end{abstract}

\section{Introduction}
The defocusing Ablowitz--Ladik (AL) lattice  for the complex functions $\alpha_j(t)$, $j\in\Z$ and $t\in\R$, is   the system of nonlinear equations	
\begin{equation}
	\label{eq:AL}
	i	\dot{\alpha}_j =-(\alpha_{j+1}+\alpha_{j-1}-2\alpha_j)+\vert \alpha_j\vert ^2(\alpha_{j-1}+\alpha_{j+1})\,,
\end{equation}
where $\dot{\alpha}_j =\dfrac{d \alpha_j}{dt}$. We assume $N$-periodic boundary conditions $\alpha_{j+N}=\alpha_j$, for all $j\in \Z$. 
The AL lattice  was introduced by Ablowitz and Ladik  \cite{Ablowitz1974,Ablowitz1975} as the spatial  integrable discretization of the defocusing  cubic  nonlinear Schr\"odinger Equation (NLS) for the complex function 
$\psi(x,t)$, $x\in S^1$ and $t\in\R$:
\begin{equation}
	i \partial_t \psi(x,t) = -\partial^2_{x} \psi(x,t) +  2\lvert \psi(x,t) \rvert^2 \psi(x,t).
\end{equation}
The cubic NLS equation  was proved to be integrable by Zakharov and Shabat \cite{ZS}.

It is straightforward to verify that the quantity
\begin{equation}
	\label{K0}
	K^{(0)}:=\prod_{j=1}^{N}\left(1-\vert \alpha_j\vert ^2\right)\,, 
\end{equation}
is a constant of motion  for the AL lattice, namely $\dfrac{d}{dt}K^{(0)}=0$. This implies that if $\vert \alpha_j(0)\vert <1$ for all $j\in\Z$,  then $\vert \alpha_j(t)\vert <1$ for all $t>0$.
We  chose  the $N$-dimensional disc  $\D^N$  as  the phase space of the AL lattice,  here  $\D=\{z\in\C\,\vert \,\vert z\vert <1\}$. 
On $ \D^N$ we introduce the symplectic form  \cite{Ercolani,GekhtmanNenciu}
\begin{equation}
	\label{eq:symplectic_form}
	\omega = i\sum_{j=1}^{N}\frac{1}{\rho_j^2}\di \alpha_j\wedge\di \wo \alpha_j\,,\quad \rho_j=\sqrt{1-\vert \alpha_j\vert ^2}.	\end{equation} 
The corresponding Poisson bracket is defined for functions $f,g \in \cC^\infty(\D^N)$ as 
\begin{equation}
	\label{eq:poisson_bracket}
	\begin{split}
		\{f,g\} & = i \sum_{j=1}^{N}\rho_j ^2\left(\frac{\partial f}{\partial \wo \alpha_j}\frac{\partial g}{\partial \alpha_j} - \frac{\partial f}{\partial \alpha_j}\frac{\partial g}{\partial \wo \alpha_j}\right)\,.		\end{split}
\end{equation} 
The phase shift $\alpha_j(t)\to e^{-2 i t}\alpha_j(t)$ transforms the AL lattice
into the equation
\begin{equation}
	\label{AL2}
	i	\dot{\alpha}_j =-\rho_j^2(\alpha_{j+1}+\alpha_{j-1}),\quad  \rho_j=\sqrt{1-\vert \alpha_j\vert ^2},
\end{equation}
which we call the {\it reduced}  AL  equation.
We remark that  the quantity $ - 2\ln(K^{(0)})$ is the generator of the shift $\alpha_j(t)\to e^{-2 i t}\alpha_j(t)$, while $H_1 =  K^{(1)} + \wo{K^{(1)}}$ with 
\begin{equation}
	\label{K1}
	K^{(1)}:=- \sum_{j=1}^{N}\alpha_j\overline{\alpha}_{j+1},
\end{equation}
generates the  flow \eqref{AL2}.
The AL   equations \eqref{eq:AL}  have the Hamiltonian structure  
\begin{equation}
	\label{eq:hamiltonian}
	\dot{\alpha}_j =\{\al_j,H_{AL}\},\quad 	H_{AL}(\alpha_j,\overline{\al}_j) =  - 2\ln(K^{(0)}) + K^{(1)} + \overline{ K^{(1)}}.
\end{equation}

\paragraph{Integrability.}
As we have already said, the AL lattice was discovered  by Ablowitz and Ladik by discretizing the $2\times 2$ Zakharov-Shabat Lax pair \cite{Ablowitz1974} of the cubic  nonlinear Schr\"odinger equation.   For a comprehensive  review see \cite{Ablowitz2003}.  The integrability 
of the Ablowitz--Ladik system  has also been proved by constructing a bi-Hamiltonian structure \cite{Ercolani,KillipVisan}. A techniques to calculate the $\tau$-function correlators has been  introduced in \cite{Cafasso22}. 

Using the connection between orthogonal polynomials on the unit circle and the AL lattice,  Nenciu and Simon \cite{Nenciu2005,Simon2005}  constructed a new Lax pair for the AL lattice  { that sometimes is referred to as the big Lax pair and which put the AL equation on the same foot as the Toda lattice.}
The  link  between orthogonal and  biorthogonal polynomials on the unit circle and the  Ablowitz–Ladik hierarchy   (see also  \cite{AV}, \cite{KN2007})  is
the analogue of the celebrated link between the Toda hierarchy and orthogonal polynomials on
the real line  (see e.g. \cite{Deift}).  This link was also generalized  to the non-commutative case \cite{Cafasso} (see also \cite{Cafasso21}). 
Generalization of this construction to other integrable equations has been considered in \cite{Nijhoff2}.

{To construct the big Lax pair, }we  follow \cite{Nenciu2005,Simon2005}   and  we double the size of the chain according to the periodic boundary conditions, thus we consider a chain of $2N$ particles $\alpha_1, \ldots, \alpha_{2N}$ such that $\alpha_j = \alpha_{j+N}$ for $j=1,\ldots, N$.
Define the  $2\times2$ unitary matrix  $\Xi_j$ 		
\begin{equation}
	\Xi_j = \begin{pmatrix}
		\wo \alpha_j & \rho_j \\
		\rho_j & -\alpha_j
	\end{pmatrix}\, ,\quad j=1,\dots, 2N\, ,
\end{equation}
and the $2N\times 2N$ matrices
\begin{equation}
	\cM= \begin{pmatrix}
		-\alpha_{2N}&&&&& \rho_{2N} \\
		& \Xi_2 \\
		&& \Xi_4 \\
		&&& \ddots \\
		&&&&\Xi_{2N-2}\\
		\rho_{2N} &&&&& \wo \alpha_{2N}
	\end{pmatrix}\, ,\qquad 
	\cL = \begin{pmatrix}
		\Xi_{1} \\
		& \Xi_3 \\
		&& \ddots \\
		&&&\Xi_{2N-1}
	\end{pmatrix} \,.
\end{equation}
Now  let us define the  unitary  Lax matrix 
\begin{equation}
	\label{eq:Lax_matrix}
	\cE  = \cL \cM\,,
\end{equation}
that has the structure  of a $5$-band diagonal matrix 
\[
\begin{pmatrix}
	*&*&*&&&&&&&*\\
	*&*&*&&&&&&&*\\
	&*&*&*&*&&&&&\\
	&*&*&*&*&&&&&\\
	&&&&&\ddots&\ddots&&&\\
	&&&&&&*&*&*&*\\
	&&&&&&*&*&*&*\\
	*&&&&&&&*&*&*\\
	*&&&&&&&*&*&*\\
\end{pmatrix}\,.
\]
The matrix $\cE$ is a periodic  CMV  matrix 
(after Cantero Morales and Velsquez  \cite{Cantero2005}).
The $N$-periodic reduced AL equation \eqref{AL2}   is equivalent to  the following Lax equation for the matrix $\cE$:	
\begin{equation}
	\label{eq:Lax_pair}
	\dot \cE = i\left[\cE, \cE^+ + (\cE^+)^\dagger\right]\,,
\end{equation}	
where $^\dagger$ stands for hermitian conjugate and 
\begin{equation}
	\cE^+_{j,k} = \begin{cases}
		\frac{1}{2} \cE_{j,j} \quad j = k \\
		\cE_{j,k} \quad k = j + 1 \, \mod \, 2N \, \mbox{or} \, k = j + 2 \, \mod \, 2N  \\
		0 \quad \mbox{otherwise}.
	\end{cases}
\end{equation}	
We observe that $(\cE^+)^{\dagger}+(\cE^{\dagger})^+=\cE^{\dagger}$ and $[\cE,\cE^{\dagger}]=0$  because $\cE$  is unitary. Therefore, one can write the equation \eqref{eq:Lax_pair}
in the equivalent form
\begin{equation}
	\label{eq:Lax_pair1}
	\dot \cE = i\left[\cE, \cE^+ -(\cE^\dagger)^+\right]\,.
\end{equation}
The formulation \eqref{eq:Lax_pair} implies that the quantities 
\begin{equation}
	\label{eq:constant_motion}
	K^{(\ell)}=\frac{ \mbox{Tr}\left(\cE^\ell\right)}{2},\quad \ell=1,\dots, N-1,
\end{equation}
are constants of motion for the defocusing AL system. For example
\[
K^{(1)}=-\sum_{j=1}^N\alpha_j\overline{\alpha}_{j+1},\quad 	K^{(2)}=\sum_{j=1}^N[(\alpha_j\overline{\alpha}_{j+1})^2-2\alpha_j\overline{\alpha}_{j+2}\rho_{j+1}^2]\,.
\]
Furthermore,  $K^{(0)},	K^{(1)},\dots,	K^{(N-1)}$   are functionally independent and in  involution, showing that the $N$-periodic  AL system is 	integrable  \cite{Nenciu2005,Ablowitz1974,Ablowitz2003}.

\begin{remark}
	The quantity $2\Re(K^{(1)})$  generates the reduces AL equation \eqref{AL2}, while the quantity $-2\Im(K^{(1)}) $
	generates the flow
	\[
	\dot{\alpha}_j =(1-\vert \alpha_j\vert ^2)(\alpha_{j+1}-\alpha_j),
	\]
	which is called Schur flow. The Schur flow emerges in \cite{Ablowitz1975} as a spatial discretization of the 
	defocusing  modified   Korteweg--de Vries equation
	\[
	\partial_t f-6f\partial_xf+\partial_x^3 f=0.
	\]
	For the integration of the Schur flow and its relation to orthogonal polynomial on the unit circle see \cite{Golinskii,Simon11}.
\end{remark}
\paragraph{Generalized Gibbs Ensemble for the Ablowitz--Ladik Lattice.}
The symplectic form $\omega$ in \eqref{eq:symplectic_form} induces  on $\D^N$ the  volume form $\di\mbox{vol}=\dfrac{1}{K^{(0)}}\di^2 \balpha$,  with  $\di^2 \balpha=\prod_{j=1}^{N}(i\di \alpha_j\wedge\di \wo \alpha_j\,) $.  We observe that $\int_{\overline{\D}^N}\di\mbox{vol}=\infty$, however,  we can define the Gibbs measure with respect to the  Hamiltonian $H_{AL}$ in \eqref{eq:hamiltonian}:
\begin{equation}
	\label{Gibbs}
	\frac{1}{Z_{\beta}}e^{-\frac{\beta}{2} H_{AL}}\di\mbox{vol} = \frac{1}{Z_{\beta}}e^{\beta \Re(K^{(1)})}\prod_{j=1}^{N}(1-\vert \alpha_j\vert ^2)^{\beta-1}\di^2 \balpha,\quad\beta>0,
\end{equation}
where $Z_{\beta}=\int_{\overline{\D}^N}e^{\beta \Re(K^{(1)})}\prod_{j=1}^{N}(1-\vert \alpha_j\vert ^2)^{\beta-1}\di^2 \balpha<\infty$ is the normalizing constant.
The above probability measure is clearly invariant under the Hamiltonian flow  $\al_j(0)\to\al_j(t)$  associated to the Ablowitz--Ladik equation  \eqref{eq:AL}.

Since  the Ablowitz--Ladik lattice posses several conserved quantities \eqref{eq:constant_motion},   one  can introduce a  Generalized Gibbs Ensemble  on the phase space $\D^N$ in the following way.  Fix $\N \ni \kappa \leq N-1$ and let us define
\begin{equation}
	\label{eq:potential}
	V(z) =\sum_{m=1}^{\kappa}\eta_m \Re(z^m)\, ,
\end{equation}
where $\eta_m \in \R$ are called chemical potentials.   Then 
$$\mbox{Tr}\left(V\left(\cE\right)\right)=\sum_{m=1}^{\kappa} \eta_m (K^{(m)}+\overline{K^{(m)}} ),$$
where $K^{(m)}$ are 	the AL conserved quantities \eqref{eq:constant_motion}.
The  finite volume Generalized Gibbs measure can be written as:	
\begin{equation}
	\label{eq:GGE}
	\di \mathbb{P}_{AL}(\al_1, \ldots, \al_{N})= \frac{1}{Z^{AL}_N(V,\beta)}\prod_{j=1}^{N} \left(1-\vert \alpha_j\vert ^2\right)^{\beta-1}\exp\left(-\mbox{Tr}\left(V\left(\cE\right)\right)\right)  \di^2 \balpha\, ,
\end{equation} 
where  $Z^{AL}_N(V,\beta)$ is the partition function of the system:

\begin{equation}
	\label{eq:partition_AL}
	Z^{AL}_N(V,\beta)= \int_{\D^N} \prod_{j=1}^{N} \left(1-\vert \alpha_j\vert ^2\right)^{\beta-1}\exp\left(-\mbox{Tr}\left(V\left(\cE\right)\right)\right)  \di^2 \balpha\, .
\end{equation}
Choosing the initial   data of the Ablowitz--Ladik lattice according to the Generalized Gibbs measure \eqref{eq:GGE}, the Lax matrix $\cE$  turns into a random matrix.  
In \cite{SpohnMendl} Mendl and Sphon study the  dynamic of the Ablowitz--Ladik lattice at non-zero temperature. They study numerically   correlation functions and in particular,
introducing the density  $\delta_j=\Re(\alpha_{j+1}\overline{\alpha}_j)$, they study  the density-density correlation function
\[
\meanval{\delta_j(t)\delta_1(0)}-\meanval{\delta_j(t)}\meanval{\delta_1(0)}\,,
\]
where $\meanval{\cdot}$ is the expectation with respect to Gibbs measure \eqref{Gibbs}. They showed numerically that  density-density time correlations in thermal equilibrium have  symmetrically  located    peaks,  which  travel  in opposite directions at constant speed,  broaden ballistically  and decay as $1/t^{\gamma}$ when  $t\to\infty$, where  the scaling exponent $\gamma$ is approximately equal to one. This  behaviour is believed to be typical of integrable nonlinear systems.

More quantitative results have been obtained for  linear (integrable) systems and for the Toda lattice. It was shown in  \cite{GKMM} that the fastest peaks of the correlation functions of  harmonic oscillators with short range interactions  have  a Airy type scaling. 
Regarding nonlinear integrable systems in \cite{Spohn2019a}  Spohn  was able to connect   the  Gibbs ensemble of the Toda lattice to the  Dumitriu-Edelman $\beta$-ensemble \cite{Dumitriu2002}. In this way, the generalized Gibbs free energy of the Toda chain turns out to be related to the $\beta$-ensembles of random matrix theory in the mean-field regime \cite{Duy2018,Allez2012}.  The behaviour of the correlation functions  of the Toda chains has been derived by applying the theory of generalized  hydrodynamic \cite{Spohn_2021, Doyon_notes}. We mention also the recent work \cite{guionnet2021large}, where the authors derive a large deviation principle for the mean density of states for the Generalized Gibbs measure of the Toda lattice.

\section{Statement of the results}
In this manuscript we derive the mean density of states $\mu^\beta_{AL}$ of the random Lax matrix $\mathcal{E}$  sampled according to  generalized Gibbs measure \eqref{eq:GGE} and we determine the free energy of the AL generalized Gibbs ensemble $$F_{AL}(V,\beta)=\lim_{N\to\infty}\frac{1}{N} \log	Z^{AL}_N(V,\beta). $$   This is achieved 
by connecting the generalized Gibbs ensemble of the Ablowitz--Ladik lattice  to  the  Killip-Nenciu \cite{Killip2004} matrix Circular $\beta$-ensemble at high-temperature investigated  by Hardy and Lambert \cite{Hardy2020}. 
Further connections between discrete integrable systems with  Gibbs measure initial data and classes of random matrices has been explored in \cite{GGGM}.  For connections between integrable PDEs and random objects see \cite{Bothner}.

Let  $\mathcal{M}(\T)$  be the space of probability measures on the torus $\T=[-\pi,\pi]$ and for $\mu\in\mathcal{M}(\T)$  let us consider the functional
\begin{equation}
	\label{functional_intro}
	\begin{split}
		\cF^{(V,\beta)}(\mu) & = 2\int\limits_{\T} V(\theta) \mu(\theta) \di \theta - \beta \int\int\limits_{\hskip-0.2cm\T\times \T} \ln\sin\left(\frac{\vert \theta-\phi\vert }{2}\right)\mu(\theta)\mu(\phi)\di \theta\di \phi  \\ & + \int\limits_{\T}\ln\left(\mu(\theta)\right) \mu(\theta)\di \theta+\ln(2\pi)\,.
	\end{split}
\end{equation}		
\begin{remark}
	Here and below, we make an abuse of notation by    denoting the potential $V(z)=V(e^{i\theta})$   simply by $V(\theta)$.
\end{remark}

For sufficiently regular potential $V(\theta)$, the functional  \eqref{functional_intro} has a unique minimizer  $\mu_{HT}^{\beta}(\di\theta)=\mu_{HT}^{\beta}(\theta)\di\theta$, \cite{SaffBook}, that describes the density of states of the Circular $\beta$-ensemble at high-temperature \cite{Hardy2020}.   For finite $\beta$ and smooth potentials $V(\theta)$,  it has been shown 
by Hardy and Lambert in \cite{Hardy2020} that  the minimizer $\mu_{HT}^{\beta}(\di\theta)$ has a smooth density and its support is the whole torus $\T$. 

\begin{theorem}[First Main theorem]
	\label{THM:RELATION}
	Consider $\beta >0$ and a  smooth  potential $V$ as in \eqref{eq:potential}  on the unit circle $\T$. The mean density of states $\mu_{AL}^\beta(\di \theta):=\mu_{AL}^\beta(\theta)\di \theta$ of the Ablowitz--Ladik  Lax  matrix  $\cE$ in  \eqref{eq:Lax_matrix}   endowed with the probability \eqref{eq:GGE}  is   absolutely continuous with respect to the Lebesgue measure and takes the form  
	\begin{equation}
		\label{eq:ass_cont}
		\mu_{AL}^\beta(\theta) = \partial_\beta \left(\beta \mu_{HT}^\beta(\theta)\right) \; a.e.,
	\end{equation}
	where $\mu_{HT}^\beta$ is the  unique  minimizer of the functional \eqref{functional_intro}.
\end{theorem}
To prove the above theorem  we derive a relation (see Proposition~\ref{PROP:MOMENT_RELATION}) 
between the free energy of the $\beta$-ensembles at high-temperature, namely 	the minimum value of the minimizer \eqref{functional_intro} 
$$F_{HT}(V,\beta):=
\cF^{(V,\beta)}(\mu^\beta_{HT}),$$  and the  free energy $F_{AL}(V,\beta)$ of the AL lattice:
\[
F_{AL}(V,\beta)= \partial_\beta \left(\beta F_{HT}(V,\beta)\right) + \ln(2).
\]
Such relation is obtained via transfer operator techniques.

The particular case $V(\theta)=2\eta \cos\theta$  corresponds to the  free energy  associated to the  AL equation \eqref{eq:AL},  and we show that  the  minimizer  of the functional 	\eqref{functional_intro} is obtained via a particular solution of the Double Confluent Heun (DCH)  equation.
\begin{theorem}[Second main theorem]
	\label{THM:MEAN_DENSITY}
	Fix $\beta >0$ and let $V(\theta)=\eta\cos\theta$, where $\eta$ is a real parameter. There exists  $\varepsilon >0$ such that for all $\eta \in (-\varepsilon,\varepsilon)$, the minimizer
	$\mu_{HT}^{\beta}(\di \theta)=\mu^{\beta}_{HT}(\theta)\di \theta$ of the functional \eqref{functional_intro} takes the form
	
	\begin{equation}
		\label{mu_intro} 
		\mu^{\beta}_{HT}(\theta)= 
		\frac{1}{2\pi} + \frac{1}{\pi \beta }\Re\left( \frac{zv'(z)}{v(z)}\right)_{\big\vert z=e^{i\theta}}\, ,
	\end{equation}
	where $v(z)$ is the unique solution (up to a multiplicative non-zero constant)  of the Double Confluent Heun (DCH) equation
	\begin{equation}
		\label{eq:heun}
		z^2 v''(z) + \left(-\eta + z (\beta +1) + \eta z^2\right)v'(z) +\eta\beta (z +\lambda) v(z) =0\,
	\end{equation}
	analytic   for  $\vert z\vert \leq r$ with  $r\geq 1$. Such solution is differentiable in the parameter $\eta$ and $\beta$. The  parameter 	$\lambda=\lambda(\eta,\beta)$ in  \eqref{eq:heun}   is determined for   $ \eta \in (-\varepsilon,\varepsilon)$ by the solution of the equation
	\begin{equation}
		\lambda (R_1)_{11}+  \frac{\eta}{\beta +1}(R_1)_{21}=0,\quad 		
	\end{equation}
	with   the condition $\lambda(\eta=0,\beta)=0$. In the above expression  $(R_1)_{jk}$ is the $jk$ entry of the matrix $R_1$ which is defined by the infinite product
	\[
	R_1 =M_1M_2\dots M_k\dots, \quad  M_k = \begin{pmatrix}
		1 + \frac{\lambda\beta \eta}{k(k+\beta)} & \frac{\eta^2}{k(k+\beta+1)}\\
		1 & 0
	\end{pmatrix}\,.
	\]
\end{theorem}
We remark that the solution of the double confluent Heun equation has generically an essential singularity at $z=0$ and $z=\infty$,
and one needs to tune the \textit{accessory parameter}   $\lambda$ to obtain an analytic solution, for a review see \cite{Ronveaux}. In our derivation of  \eqref{eq:heun}  the parameter $\lambda$ coincides with the first moment of the measure $\mu(\theta)$, namely $\lambda=\int_{\T}\mu(\theta)e^{i\theta}d\theta$. It is  a transcendental function of $\beta$ and $\eta$ and  it   is related to  the Painlev\'e III equation \cite{FIKN,Lisovyy}.
\begin{remark}
	Under the change of variable
	\[
	v(z)=\exp\left(-\frac{\eta}{2}\left(z+\frac{1}{z}\right)\right)z^{-\frac{\beta+1}{2}}f(z)\,,
	\]
	the DCH equation \eqref{eq:heun} takes the form  of a Schr\"odinger equation 
	\[
	f''(z)+q(z)f(z)=0,
	\]
	with   potential  $q(z)$  singular at the origin
	\[ q(z)=\frac{1}{z^2}\left(
	\eta\frac{\beta-1}{2}\left(z+\frac{1}{z}\right)-\frac{\eta^2}{4}\frac{(z^2-1)^2}{z^2} -\frac{1}{4}(\beta^2+4\beta+3)+\eta\beta\lambda\right).
	\]
\end{remark}
\begin{remark}
	For the case $V=0$  it was shown in \cite{Hardy2020} that the minimizer of the functional \eqref{eq:functional} is the uniform measure on the circle, while for the case  $V(\theta)= \beta V(\theta) $
	and $\beta\to\infty$ the minimizer of \eqref{functional_intro} was considered in \cite{MM}. The particular case $  V(\theta)\to\beta\eta\cos\theta$ and $\beta\to\infty$  has first been considered by Gross--Witten  \cite{Gross_Witten} and Baik--Deift--Johansson \cite{BDJ}.
	The measure \eqref{mu_intro}  in Theorem~\ref{THM:MEAN_DENSITY} generalizes the result of  Gross and Witten  \cite{Gross_Witten} and Baik--Deift--Johansson \cite{BDJ}   to the high-temperature regime (see  Remark~\ref{rem_Deift}).
\end{remark}

This manuscript is organized as follows. In section~\ref{sec2} we introduce the Circular $\beta$ ensemble and its high-temperature limit.  Then we review results in the literature on Circular $\beta$ ensemble and we derive some technical results needed to prove our main theorems.
In section~\ref{sec3} we  prove our first main theorem,  namely Theorem~\ref{THM:RELATION} and  in section~\ref{sec4} we prove Theorem~\ref{THM:MEAN_DENSITY}.
Finally, the most technical parts of our arguments are deferred to the appendices.

\section{Circular $\beta$  Ensemble at high-temperature}
\label{sec2}
The  Circular  Ensemble  at temperature $\tilde{\beta}^{-1}$   is a system of  $N$ identical particles  on the  one-dimensional  torus $\T=[-\pi,\pi]$  with distribution 
\begin{equation}
	\label{eq:CBU_joint_density}
	\di \mathbb{P}_{\tilde{\beta}} (\theta_1, \ldots, \theta_{N}) = \frac{1}{Z^{C\tilde{\beta} E}_N} \prod_{j < \ell}\vert e^{i\theta_\ell} - e^{i\theta_j}\vert ^{\wt \beta} \di \btheta,\quad \di \btheta=\di \theta_1\dots\di\theta_N \, , \end{equation}
where   $Z_N^{C\tilde{\beta} E}>0$ is the  norming constant, or  partition function of the system.
For $\tilde{\beta}=1,2,4$   Dyson observed that the above measure corresponds to the   eigenvalue  distribution of circular orthogonal ensemble (COE), circular  unitary ensemble  (CUE) and circular symplectic ensemble  (CSE)  of random matrix ensembles  (see e.g. \cite{Mehta2004book, Forrester2010book}). 
For general $\tilde{\beta}>0$, Killip and Nenciu  proved that the Circular $\beta$  Ensemble  can be associated to the eigenvalue distribution   of a random
sparse matrix, the so-called   CMV matrix, after Cantero, Moral, Vel\'azquez \cite{Cantero2005}. To state their result, we need the following definition.

\begin{definition}
	\label{def:theta}
	A complex random variable $X$ with values on the unit disk $\D$ is $\Theta_\nu$-distributed ($\nu>1$) if
	\begin{equation}
		\meanval{f(X)} = \frac{\nu-1}{2\pi} \int_{\overline{\D}} f(z)(1-\vert z\vert ^2)^{\frac{\nu-3}{2}}\di^2z\, .
	\end{equation}
	for any  measurable function  $f\,:\,\D\to\C$. When $\nu = 1$,  $\Theta_1$ is the uniform distribution on the unit circle $S^1$. 
\end{definition}
We recall that for $\N\ni\nu\geq 2$,  such measure has the following geometrical interpretation:   if $\mathbf{u}=(u_1,u_2,\dots,u_{\nu+1})$ is chosen  at random according to the surface measure  on  the unit sphere $S^\nu$  in $\mathbb{R}^{\nu+1}$,   then $u_1+i u_2$ is $\Theta_\nu-$distributed.
We can now state the result of Killip-Nenciu.

\begin{theorem}[cf. \cite{Killip2004} Theorem 1]
	\label{thm:Killipenciu}
	Consider the block diagonal $N\times N$ matrices
	\begin{equation}
		\label{eq:LME}
		M = \mbox{diag}\left(\Xi_1,\Xi_3,\Xi_{5} \ldots,\right) \quad \mbox{ and } \quad L = \mbox{diag}\left(\Xi_{0},\Xi_2,\Xi_4, \ldots\right)\,,
	\end{equation}
	where the block $\Xi_j$, $j=1,\dots, N-1$, takes the form 
	\begin{equation}
		\label{eq:xi_def}
		\Xi_j = \begin{pmatrix}
			\wo \alpha_j & \rho_j \\
			\rho_j & -\alpha_j
		\end{pmatrix}\, ,\;\;\rho_j = \sqrt{1-\vert \alpha_j\vert ^2}, 
	\end{equation}	
	while $\Xi_{0} = (1)$ and $\Xi_{N} = (\wo \alpha_{N})$ are $1\times 1$ matrices. 
	Define the  $N \times N$ sparse matrix 
	\begin{equation}
		\label{E}
		E = LM,
	\end{equation}
	and suppose that  the entries $\alpha_j $ are independent complex random variables with $\alpha_j\sim \Theta_{\wt \beta(N-j) +1 }$ 	 for $1\leq j\leq N-1$ and  $\alpha_{N}$ is  uniformly distributed on the unit circle.
	Then the eigenvalues of $E$ are distributed according to the Circular  Ensemble  \eqref{eq:CBU_joint_density} at temperature $\tilde{\beta}^{-1}$.

\end{theorem}
We observe that each of the matrices $\Xi_j$ is unitary, and so are the matrices $L$ and $M$. As a result, the eigenvalues of $E$  clearly  lie on the unit circle.
The matrix $E$   is a $5$-diagonal unitary matrix that takes the form
\vspace{5pt}
{
	\begin{equation*}
		E=\begin{pmatrix}
			\bar{\alpha}_1&\rho_1\bar{\alpha}_2&\rho_1\rho_2&&&&&&\\
			\rho_1&-\alpha_1\bar{\alpha}_2&-\alpha_1\rho_2&&&&&&\\
			&\rho_2\bar{\alpha}_3&-\alpha_2\bar{\alpha}_3&\rho_3\bar{\alpha}_4&\rho_3\rho_4&&&&\\
			&\rho_2\rho_3&-\alpha_2\rho_3&-\alpha_3\bar{\alpha}_4&-\alpha_3\rho_4&&&&\\
			&&&\ddots&\ddots&\ddots&\ddots&&&\\
			&&&&&\rho_{N-3}\bar{\alpha}_{N-2}&-\alpha_{N-3}\bar{\alpha}_{N-2}&\rho_{N-2}\bar{\alpha}_{N-1}&\rho_{N-2}\rho_{N-1}\\
			&&&&&\rho_{N-3}\rho_{N-2}&-\alpha_{N-3}\rho_{N-2}&-\alpha_{N-2}\bar{\alpha}_{N-1}&-\alpha_{N-2}\rho_{N-1}\\
			&&&&&&&\bar{\alpha}_{N}\rho_{N-1}&-\alpha_{N-1}\bar{\alpha}_N\\
		\end{pmatrix}\,.
\end{equation*}}

We are   interested in  the probability distribution  \eqref{eq:CBU_joint_density} when
\begin{itemize}
	\item  we add an external field, namely $d\theta_i\to e^{-2V(\theta_i)}d\theta_i $ with $V:\T\to \R$  a smooth potential;
	\item we consider the limit  $\wt \beta \to 0$  and  $N \to \infty$ in such a way that $\wt \beta N = 2\beta, \, \beta >0.$ Since $\wt \beta$ is interpreted as the inverse of the temperature, such limit is  called 
	{\it  high-temperature regime}.
\end{itemize}
With the above changes, we arrive to the probability distribution  of the Circular ensemble at high-temperature, and with an external potential:
\begin{equation}
	\label{eq:CBU_ht_joint_density}
	\di \mathbb{P}_{\beta}^V (\theta_1, \ldots, \theta_{N}) = \frac{1}{\cZ^{HT}_N(V,\beta)} \prod_{j < \ell}\vert e^{i\theta_\ell} - e^{i\theta_j}\vert ^{\frac{2 \beta}{N}}e^{-2\sum_{j=1}^{N} V(\theta_j)}\di \btheta \, ,
\end{equation}
where $\cZ^{HT}_N(V,\beta)$ is   the partition function of the system.
Also in this case, we can associate to the above probability distribution a random CMV matrix.  The lemma below has probably already appeared in the literature, 
but for completeness we provide the proof.
\begin{lemma}
	
	Let $E$ be the CMV matrix \eqref{E}. Consider the block $2N\times 2N$ matrix 
	
	\begin{equation}
		\label{eq:doubble_E}
		\wt E = \mbox{diag}(E,E)\, ,
	\end{equation}
	whose entries are distributed according to	
	\begin{equation}
		\label{eq:HT_alpha}
		\di \mathbb{P}(\al_1, \ldots, \al_{N})  = \frac{1}{Z^{HT}_N(V,\beta)} \prod_{j=1}^{N-1}\left( 1 - \vert \alpha_j\vert ^2\right)^{\beta\left(1 -\frac{j}{N}\right) -1} e^{-\mbox{Tr}(V(\wt E))}
		\prod_{j=1}^{N-1}\di^2 \alpha_j\frac{\di \alpha_N}{i\alpha_N}\,.
	\end{equation}
	Then the eigenvalues  of $\wt E$ are all double, they   lie  on the  unit  circle and are distributed according to \eqref{eq:CBU_ht_joint_density}.
	
	\noindent Moreover 
	\begin{equation}
		\label{eq:ZHT_relation}
		Z^{HT}_N(V,\beta) = 2^{1-N}\frac{\Gamma\left(\frac{\beta}{N} \right)^N}{\Gamma
			(\beta)}\cZ^{HT}_N(V,\beta)\,,
	\end{equation}
	where $\cZ^{HT}_N(V,\beta)$ is the norming constant of the probability distribution
	\eqref{eq:CBU_ht_joint_density} and $Z^{HT}_N(V,\beta) $  is the norming constant of the probability distribution \eqref{eq:HT_alpha}.
\end{lemma}
\begin{proof}	
	First, we notice that the eigenvalues of $\wt E$ are all double, since it is a block diagonal matrix with two identical blocks.
	
	We consider the change of variables $\alpha_N \to e^{i \varphi}$, thus \eqref{eq:HT_alpha} becomes:
	\begin{equation}
		\label{eq:HT_mod}
		\di \mathbb{P}(\al_1, \ldots,\alpha_{N-1}, \varphi)  = \frac{ \prod_{j=1}^{N-1}\left( 1 - \vert \alpha_j\vert ^2\right)^{\beta\left(1 -\frac{j}{N}\right) -1} e^{-\mbox{Tr}(V(\wt E))}
			\prod_{j=1}^{N-1}\di^2 \alpha_j\di \varphi}{Z^{HT}_N(V,\beta)}\,.
	\end{equation}
	Now, let $e^{i\theta_1}, \ldots,e^{i\theta_N}$ be the eigenvalues of the CMV matrix $E$ endowed with probability \eqref{eq:HT_alpha}, and let $q_1, \ldots,q_N$ be the entries of the first row of the unitary matrix $Q$ such that $Q^\dagger \Theta Q = E$  where $\Theta=\mbox{Diag}(e^{i\theta_1}, \ldots,e^{i\theta_N})$ and $\sum_{k=1}^N\vert q_k\vert ^2=1$. We introduce the variable $\gamma_j = \vert q_j\vert ^2$ for $j=1,\ldots, N$, then
	from \cite{Killip2004} (Lemma 4.1, and  relation (4.14) in Proposition 4.2) we have 		 
	\begin{align}
		\label{eq:1}
		& \vert \Delta(e^{i\btheta})\vert ^2 \prod_{j=1}^{N}\gamma_j = \prod_{j=1}^{N-1}\left( 1 - \vert \alpha_j\vert ^2\right)^{(N-j)}\, , \\ 
		\label{eq:2}
		&\left\vert  \frac{\partial\left(\alpha_1, \ldots, \alpha_{N-1}, \varphi\right)}{\partial(\btheta,\bgamma)}\right\vert   = 2^{1-N} \frac{\prod_{j=1}^{N-1}\left( 1 - \vert \alpha_j\vert ^2\right)}{\prod_{j=1}^{N}\gamma_j}\, ,
	\end{align}
	here $\bgamma = (\gamma_1,\ldots,\gamma_{N-1})$, and $\Delta(e^{i\btheta}) = \prod_{j<\ell}\left(e^{i\theta_j} - e^{i\theta_\ell}\right)$. 
	Applying the previous equalities to \eqref{eq:HT_mod} we derive
	
	\begin{equation}
		\begin{split}
			&\di \mathbb{P}(\alpha_1,\ldots, \alpha_{N-1}, \varphi))  = \frac{e^{-\mbox{Tr}(V(\wt E))}}{Z^{HT}_N(V,\beta)}\di\varphi \prod_{j=1}^{N-1}\left( 1 - \vert \alpha_j\vert ^2\right)^{\beta\left(1 -\frac{j}{N} \right) -1} \di \alpha_j \di\overline{\alpha}_j  \\ & \qquad\stackrel{\eqref{eq:2}}{=} \frac{1}{Z^{HT}_N(V,\beta)}\frac{2^{1-N}}{\prod_{j=1}^{N}\gamma_j} \prod_{j=1}^{N-1}\left( 1 - \vert \alpha_j\vert ^2\right)^{\beta\left(1 -\frac{j}{N} \right)} e^{-2\sum_{j=1}^N V(e^{i\theta_j})} \di \boldsymbol{\theta}\di \bgamma \\
			&\qquad \stackrel{\eqref{eq:1}}{= }\frac{1}{Z^{HT}_N(V,\beta)} 2^{1-N}\vert \Delta(e^{i\btheta})\vert ^{\frac{2\beta}{N}} \prod_{j=1}^N \gamma_j^{\frac{\beta}{N}-1} e^{-2\sum_j V(e^{i\theta_j})} 
			\di \boldsymbol{\theta}\di \bgamma\,.
		\end{split}
	\end{equation}
	Thus, we deduce  the relation
	\begin{equation}
		Z^{HT}_N(V,\beta) = 2^{1-N}\cZ^{HT}_N(V,\beta)\, \int_{\Delta} \prod_{j=1}^N \gamma_j^{\frac{\beta}{N} - 1} \di \gamma_1\ldots\di \gamma_{N-1}\,, 
	\end{equation}
	here $\Delta$ is the simplex $\sum_{j=1}^N \gamma_j = 1$. 
	The above integral  is a well-known Dirichlet integral that can be computed explicitly (see \cite[Lemma 4.4]{Killip2004})
	\begin{equation}
		\int_{\Delta} \prod_{j=1}^N \gamma^{\frac{\beta}{N} - 1} \di \gamma_1\ldots\di \gamma_{N-1}= \frac{\Gamma\left(\frac{\beta}{N} \right)^N}{\Gamma
			(\beta)}\, ,
	\end{equation}
	proving  \eqref{eq:ZHT_relation}.
\end{proof}
Let $e^{i\theta_1},\dots, e^{i\theta_N}$ be the  double  eigenvalues of the  CMV Matrix $\tilde{E}$ in \eqref{eq:doubble_E}, whose entries are distributed according to \eqref{eq:HT_mod}.  The {\it empirical measure} is the  random probability measure 	
\begin{equation}
	\label{empirical}
	\mu_N=\dfrac{1}{N}\sum_{j=1}^N\delta_{e^{i\theta_j}}.
\end{equation}
The mean density of state $\mu^\beta_{HT}$ is defined as the non-random probability measure such that 
\begin{equation}
	\label{eq:mean_density_general}
	\int_\T f(\theta) \mu^\beta_{HT}(\di \theta) = \lim_{N\to \infty} \meanval{\int_\T f(\theta)\mu_N(\di \theta)}\, ,
\end{equation}
for all  continuous function $f$ on the torus $\T$, and the expected value is taken
with respect to \eqref{eq:HT_alpha}.   In order to discuss the large $N$ limit of $\mu_N$  we have to introduce   several   quantities.
Let $\mathcal{M}(\T)$ be the  set  of  probability measures on the one-dimensional torus  $\T$  and for $\mu\in \mathcal{M}(\T)$ we consider the logarithmic energy \cite{SaffBook}
\[
\mathit{E}(\mu):=  \int\int\limits_{\T\times \T} \ln\left\vert \sin\left(\frac{\theta-\phi}{2}\right)\right\vert ^{-1}\mu(\di \theta)\mu(\di \phi )\,.
\]
We define   the relative entropy  $K(\mu\vert \mu_0)$ of $\mu$  with respect to  $\mu_0(\di \theta)=\dfrac{\di \theta}{2\pi}$ as
\[
K(\mu\vert \mu_0):=\int_{\T}\log\left(\frac{\mu}{\mu_0}\right)\mu(\di \theta),
\]
when $\mu$ is absolutely continuous with respect to $\mu_0$ and  otherwise $K(\mu\vert \mu_0):= +\infty$.
The  relevant  functional is 
\[
\cF^{(V,\beta)}(\mu):=\beta \mathit{E}(\mu)+K(\mu\vert \mu_0)+2\int_{\T}V(\theta)\mu(\di \theta).
\]
When $\cF^{(V,\beta)}(\mu)$ is finite, it follows that $\mu$ is absolutely continuous with respect to the Lebesgue measure $\mu_0$ and we can write $\mu(\di \theta)=\mu(\theta)\di \theta$. 
We denote by $C^{n,1}(\T)$ with $n=0, 1,2,\dots$ the space  of $n$-times differentiable functions  whose $n$-derivative is also  Lipschitz continuous. 

The following result describes the limiting measure $\mu^\beta_{HT}$ in \eqref{eq:mean_density_general}  for the circular  $\beta$-ensembles at high temperature.

\begin{theorem}(cf. \cite[Proposition 2.1 and 2.5]{Hardy2020})
	\label{thm:lambert}
	Let $\mathcal{M}(\T)$ be the  set  of  probability measures on the one-dimensional torus   and 
	$V \, :\, \T \to \R $ be  a   measurable  and bounded function.  For  any $\beta > 0$ consider the  functional $\cF^{(V,\beta)}:\mathcal{M}(\T)\to [0,\infty]$ 
	\begin{equation}
		\label{eq:functional}
		\cF^{(V,\beta)}(\mu) = 2\int_{\T} V(\theta) \mu(\theta) \di \theta +\beta\mathit{E}(\mu)+ \int_{\T}\ln\left(\mu(\theta)\right) \mu(\theta)\di \theta+\ln(2\pi)\,.
	\end{equation}	
	Then 
	\begin{enumerate}
		\item[(a)] the functional $ \cF^{(V,\beta)}(\mu)$  has a unique minimizer $\mu^{\beta}_{HT}(\di\theta)=\mu^{\beta}_{HT}(\theta)\di \theta$ in  $\mathcal{M}(\T)$;
		\item[(b)] $\mu^{\beta}_{HT}$ is absolutely continuous  with respect to the Lebesgue measure  and  there is $0<\delta<1 $ such that 
		\[
		\delta\leq \frac{\mu^{\beta}_{HT}(\theta)}{2\pi}\leq \delta^{-1},\quad a.e.\,;
		\]						\item[(c)] if $V=0$, then $\mu^{\beta}_{HT}(\di \theta) = \frac{1}{2\pi}\di \theta$;
		\item[(d)] 	 if $V \in C^{m,1}(\T)$, then  $\mu^{\beta}_{HT} \in C^{m,1}(\T)$;
		\item[(d)]  the  empirical measure  $\mu_N$  in \eqref{empirical}  converges  weakly and  almost surely to the measure  $\mu^{\beta}_{HT}$ as $N\to\infty$.
		
	\end{enumerate}
\end{theorem}
From the above theorem when the potential $V$ is at least $C^{2,1}(\T)$   the minimizer  of the functional $\cF^{(V,\beta)}$ is characterized by the 
Euler-Lagrange equations
\begin{equation}
	\label{eq:VarEq0}
	\dfrac{\delta \cF^{(V,\beta)}}{\delta \mu}=2V(\theta) -2\beta \int_{\T} \ln\sin\left(\frac{\vert \theta-\phi
		\vert }{2}\right)\mu(\phi)\di \phi  +\ln \mu(\theta)+1=C(V,\beta)\, 
\end{equation}		
where $C(V,\beta)$ is a constant  in $\theta$.
Below we derive further properties of the minimizer $\mu^{\beta}_{HT}$ following \cite{guionnet2021large}. 

\begin{lemma}
	\label{Lemma2.5}
	For any $V(z)$ as in \eqref{eq:potential}, any $\beta>0$ the following holds
	\begin{itemize}
		\item[a)] The map $\beta \to \inf \left( \cF^{(V,\beta)}(\mu)\right)$ is Lipschitz;
		\item[b)] The maps $t \to  \inf \left( \cF^{(V+t\Re(z^m),\beta)}(\mu)\right),\, t \to  \inf \left( \cF^{(V+t\Im(z^m),\beta)}(\mu)\right)$ are Lipschitz;
		\item[c)]Let $D$ be the distance on $\cM(\T)$ given by
		\begin{equation}
			\label{eq:distance}
			\begin{split}
				D(\mu, \mu')  & = \left( - \int\int\ln\left\vert \sin \left( \frac{\theta - \phi}{2}\right)\right \vert (\mu - \mu')(\di\theta)(\mu - \mu')(\di \phi)\right)^{1/2}\\ 
				& = \sqrt{\sum_{k\geq 1} \frac{1}{k}\left\vert \wh \mu_k - \wh \mu'_k \right \vert^2}\,,
			\end{split}
		\end{equation}
		where $\wh \mu_k = \int_\T e^{ik\theta} \mu(\di \theta)$, and we recall that $\wh{\log(\vert x\vert )} = \sum_{k\geq 1}k^{-1}$ in distributional sense.

		Then for any $\varepsilon>0$ there exists a finite constant $C_\varepsilon$ such that for all $\beta,\beta'>\varepsilon$
		\begin{equation}
			\label{Distance}
			D(\mu_{HT}^\beta, \mu_{HT}^{\beta'}) \leq C_\varepsilon \left\vert \beta - \beta' \right \vert\,.
		\end{equation}
	\end{itemize}
\end{lemma}

\begin{remark}
	\label{rem:lip_moments}
	
	For a real valued  function   $f\in L^2(\T)$ with derivative in $L^2(\T)$  we define $\vert \vert  f\vert \vert _{\frac{1}{2}} = \sqrt{\sum_{k\geq 1} k \vert \wh f_k\vert ^2}<\infty$. So, for any measure $\nu$ with zero mass we deduce that 
	
	\begin{equation}
		\label{eq:CS}
		\begin{split}
			\left\vert    \int_{\T} f(\theta) \nu(\di\theta) \right\vert ^2& = \left\vert \sum_{k\ne 0} \wh f_k \overline{\wh \nu_k} \right\vert ^2= \left\vert \sum_{k\ne 0} \sqrt{\vert k\vert }\wh f_k \frac{\overline{\wh \nu_k}}{\sqrt{\vert k\vert }}\right\vert ^2\leq \left \vert \sum_{k\ne 0}\vert k\vert \vert \wh f_k\vert ^2\right \vert \left \vert \sum_{k\ne 0} \frac{\vert \wh \nu_k\vert ^2}{\vert k\vert }\right \vert \\ & \leq 4  \vert \vert f\vert \vert _{\frac{1}{2}}^2 D(\nu, 0)^2
		\end{split}
	\end{equation}
	where in the first inequality we use  Cauchy-Schwartz inequality and in the second one we plug in \eqref{eq:distance}.
	
	Combining \eqref{Distance} and \eqref{eq:CS} we conclude that  for any  real valued function $f$ with finite $\vert \vert f\vert \vert _{\frac{1}{2}}$ norm,  the map $\beta \to \int_{\T}f\di \mu_{HT}^\beta(\di \theta)$ is Lipschitz for $\beta > 0$. As a consequence, the moments of $\mu_{HT}^\beta$ are almost surely differentiable with respect to $\beta$.
\end{remark}
\begin{proof}[Proof of Lemma~\ref{Lemma2.5}]
	The proof follows the lines of the corresponding one in \cite{guionnet2021large}.
	To prove points a) and b)  we exploit the same  ideas, thus we restrict  to point a).
	
	For all $\beta > \beta' > 0$ we have
	
	\[
	\cF^{(V,\beta)}(\mu_{HT}^{\beta})\leq	\cF^{(V,\beta)}(\mu_{HT}^{\beta'})=	\cF^{(V,\beta')}(\mu_{HT}^{\beta'})+(\beta -\beta')\mathit{E}(\mu_{HT}^{\beta'})
	\]
	and 
	\begin{equation}
		\label{Fe-1}
		\cF^{(V,\beta')}(\mu_{HT}^{\beta'})\leq 	\cF^{(V,\beta')}(\mu_{HT}^{\beta})=	\cF^{(V,\beta)}(\mu_{HT}^{\beta})+(\beta' - \beta)\mathit{E}(\mu_{HT}^{\beta})\,,
	\end{equation}
	so that
	\begin{equation}
		\label{Fe0}
		(\beta - \beta')\mathit{E}(\mu_{HT}^{\beta})\leq	\cF^{(V,\beta)}(\mu_{HT}^{\beta})-\cF^{(V,\beta')}(\mu_{HT}^{\beta'})\leq (\beta - \beta')\mathit{E}(\mu_{HT}^{\beta'})\,.
	\end{equation}
	Since $\mathit{E}(\mu_{HT}^{\beta}),$ and $\mathit{E}(\mu_{HT}^{\beta'})$ are finite we obtain the claim.
	
	We now move to the proof of point $c)$. Setting  $\Delta \mu = \mu_{HT}^{\beta} - \mu_{HT}^{\beta'}$   we deduce that
	
	\begin{equation}
		\begin{split}
			0 &\geq \cF^{(\beta, V)}(\mu_{HT}^{\beta}) - \cF^{(\beta, V)}(\mu_{HT}^{\beta'}) \\
			& = 2\int_{\T} V(\theta) \Delta\mu(\di\theta) -2\beta \int_{\T\times \T} \ln \left \vert \sin\left( \frac{\theta - \phi}{2}\right) \right \vert\mu_{HT}^{\beta'}(\di\theta) \Delta\mu(\di\phi) \\ & -\beta \int_{\T\times \T} \ln \left \vert \sin\left( \frac{\theta - \phi}{2}\right) \right \vert\Delta\mu(\di\theta) \Delta\mu(\di\phi)+ \int_{\T} \ln(\mu_{HT}^\beta(\theta)) \mu_{HT}^\beta(\di\theta) \\ & -  \int_{\T} \ln(\mu_{HT}^{\beta'}(\theta)) \mu_{HT}^{\beta'}(\di\theta)\\
			& = \int_{\T} \ln\left( \frac{\mu_{HT}^\beta(\theta)}{\mu_{HT}^{\beta'}(\theta)}\right) \mu_{HT}^\beta(\di\theta) + 2(\beta'-\beta)\int_{\T\times \T} \ln \left \vert \sin\left( \frac{\theta - \phi}{2}\right) \right \vert\mu_{HT}^{\beta'}(\di\theta) \Delta\mu(\di\phi) \\ & -\beta \int_{\T\times \T} \ln \left \vert \sin\left( \frac{\theta - \phi}{2}\right) \right \vert\Delta\mu(\di\theta) \Delta\mu(\di\phi)\,,
		\end{split}
	\end{equation}
	where in the second identity we used \eqref{eq:VarEq0}. Since $\int_{\T}\ln\left( \frac{\mu_{HT}^\beta(\theta)}{\mu_{HT}^{\beta'}(\theta)}\right) \mu_{HT}^\beta(\di\theta) \leq 0 $ by Jensen's inequality, we deduce that
	\begin{equation}
		\label{eq:first_est}
		\beta D\left(\mu_{HT}^{\beta}, \mu_{HT}^{\beta'}\right)^2 \leq 2(\beta - \beta')\int_{\T\times \T} \ln \left \vert \sin\left( \frac{\theta - \phi}{2}\right) \right \vert\mu_{HT}^{\beta'}(\di\theta) \Delta\mu(\di\phi)\,.
	\end{equation}
	Following \cite{guionnet2021large}, we introduce a new probability measure $\nu \in \cM(\T)$ in the previous expression, so that	
	\begin{equation}
		\label{eq:D1}
		\begin{split}
			\beta D\left(\mu_{HT}^{\beta'}, \mu_{HT}^{\beta'}\right)^2 &\leq 2(\beta - \beta')\int_{\T\times \T} \ln \left \vert \sin\left( \frac{\theta - \phi}{2}\right) \right \vert\left(\mu_{HT}^{\beta'}-\nu\right)(\di\theta) \Delta\mu(\di\phi) \\ 
			& + 2(\beta - \beta')\int_{\T\times \T} \ln \left \vert \sin\left( \frac{\theta - \phi}{2}\right) \right \vert\nu(\di\theta) \Delta\mu(\di\phi)\,.
		\end{split}
	\end{equation}
	We chose $\nu$ in such a way that the function $g_\nu(\phi) = \int_{\T} \ln \left \vert \sin\left( \frac{\theta - \phi}{2}\right) \right \vert\nu(\di\theta)$ is in $L^2(\T)$ with derivative in $L^2(\T)$. With this choice of $\nu$ and applying \eqref{eq:CS} we conclude that there exists a constant $c$ such that
	
	\begin{equation}
		\label{eq:D2}
		\left \vert \int_{\T\times \T} \ln \left \vert \sin\left( \frac{\theta - \phi}{2}\right) \right \vert\nu(\di\theta) \Delta\mu(\di\phi) \right\vert \leq c D(\mu_{HT}^\beta, \mu_{HT}^{\beta'})\,.
	\end{equation}
	Next, taking the Fourier transform and apply again the Cauchy-Schwartz inequality as in \eqref{eq:CS} we obtain 
	\begin{equation}
		\label{eq:D3}
		\begin{split}
			\left \vert \int_{\T\times \T} \ln \left \vert \sin\left( \frac{\theta - \phi}{2}\right) \right \vert\left(\mu_{HT}^{\beta'}-\nu\right)(\di\theta) \Delta\mu(\di\phi) \right \vert & = \left \vert \sum_{k\geq 1} \frac{1}{k}\wh{\left( \mu_{HT}^{\beta'}-\nu\right)}_k \overline{\wh{\Delta \mu_k}}\right \vert\\ & \leq D(\mu_{HT}^{\beta'},\nu)D(\mu_{HT}^{\beta'}, \mu_{HT}^{\beta})\,,
		\end{split}
	\end{equation}
	since $D(\mu_{HT}^{\beta'},\nu)$ is bounded. Combining  \eqref{eq:D1}, \eqref{eq:D2}  and \eqref{eq:D3}  we conclude that there exists a constant $c_0$ such that
	\begin{equation}
		D(\mu_{HT}^\beta,\mu_{HT}^{\beta'})\leq \frac{c_0}{\beta}(\beta - \beta')\,,
	\end{equation}
	from which \eqref{Distance}   follows.
	
\end{proof}

For convenience, we define  $F_{HT}(V,\beta)$  as the value of  
the functional at the minimizer, namely 
\begin{equation}
	\label{F_min}
	F_{HT}(V,\beta):=\cF^{(V,\beta)}(\mu^\beta_{HT}).
\end{equation}
The quantity $F_{HT}(V,\beta)$ is  referred to as free energy  of the Circular $\beta$ ensemble at high-temperature.
It is a standard result that  (see e.g. \cite{Zelada2019})
\begin{equation}
	\label{eq:free_eig}
	F_{HT}(V,\beta)=-\lim_{N\to\infty}\frac{1}{N}\log \cZ^{HT}_N(V,\beta)\,,
\end{equation}
where the partition function $\cZ^{HT}_N(V,\beta)$  of the Circular $\beta$ ensemble at high-temperature is defined in \eqref{eq:CBU_ht_joint_density}.

\begin{remark}
	\label{rem:free_energy}
	We notice that from \eqref{eq:ZHT_relation} and \eqref{eq:free_eig} we can also obtain the free energy  $F_{HT}(V,\beta)$ from the partition function $Z^{HT}_N(V,\beta)$ of the CMV matrix ensemble  \eqref{eq:HT_alpha}, namely: 
	\begin{equation}
		F_{HT}(V,\beta) = - \lim_{N\to\infty} \frac{\ln (Z^{HT}_N(V,\beta))}{N} - \ln(2)\, .
	\end{equation}
\end{remark}

The literature related to the high-temperature regime of the classical $\beta$-ensembles is quite broad. For completeness, we mention some of the results in the field. In \cite{Allez2012,Allez2013,Hardy2020,Duy2018,trinh2019beta,trinh2020beta,Forrester2021} the authors explicitly computed the mean density of states for the classical Gaussian, Laguerre, Jacobi, and Circular $\beta$ ensemble at high-temperature. In \cite{Allez2012,Allez2013,Hardy2020,Forrester2021} the densities of states are computed as a solution of some particular ordinary differential equations. On the other hand, in \cite{Duy2018,trinh2019beta,trinh2020beta}  the density of states is constructing  from the moment generating functions.  Several authors \cite{BG2015,Trinh2019,Nakano2018,Nakano2020,Lambert2021} investigated the local fluctuations of the eigenvalues, and they observed that in this regime they are described by a Poisson process. In particular, in \cite{Lambert2021}  Lambert studied the local fluctuations for general Gibbs ensembles on $N$-dimensional manifolds,  moreover he also studied the asymptotic behaviour of the maximum eigenvalue for the classical $\beta$  ensembles at high-temperature.  In \cite{Forrester2021,forrester2021highlow} the loop equations for the classical $\beta$-ensembles at high-temperature are studied, in particular in \cite{forrester2021highlow} a  duality between high and low temperature is uncovered. There are also results   for a Coulomb gas at high temperature in two dimensions  \cite{Akemann2019}. It is worth mentioning also the work \cite{Mazzuca2020}, where some new tridiagonal random matrix ensembles related to the classical $\beta$ one at high-temperature are defined.

\section{Proof of Theorem~\ref{THM:RELATION}}
\label{sec3}
The  probability distribution \eqref{eq:GGE}  of  generalized Gibbs ensemble of the Ablowitz--Ladik lattice  is very close to the probability distribution \eqref{eq:HT_alpha} of the Circular $\beta$  ensemble  at high-temperature with an external source. Indeed, the only difference between the two ensembles is the exponent of 
the terms  $\left(1-\vert \alpha_j\vert \right)$  in the probability distributions \eqref{eq:GGE} and \eqref{eq:HT_alpha} and the fact that the random matrix of the Ablowitz--Ladik lattice  is a rank $2$ perturbation of the random matrix of the circular $\beta$-ensemble.
Our first main result contained in Theorem~\ref{THM:RELATION} relates  the mean density of states of the random Lax matrix $\mathcal{E}$ of the  Ablowitz--Ladik lattice to the mean density of states  of the random matrix $E$  from the  Circular $\beta$  ensemble at high-temperature.

To prove the result, we use the {\em moment matching technique} and the following lemma.
\begin{lemma}(\cite[Lemma B.1 - B.2]{Bai2010})
	\label{lem:bay}
	Let  $\di \sigma, \di \sigma'$ be two measures   defined on  $\mathbb{T}$, with the same moment sequence $\{u^{(\ell)}\}_{\ell\geq 0}$. If 
	\begin{equation}
		\label{eq:divergence_moments}
		\lim_{\ell\to \infty}\inf\frac{(u^{(2\ell)})^\frac{1}{2\ell}}{\ell} < \infty\, ,
	\end{equation}
	then $\di \sigma = \di \sigma'$.
\end{lemma}

Next we define the {\em free energy} of the generalized Gibbs ensemble of   the  Ablowitz--Ladik lattice  at temperature $\beta^{-1}$ and in an external field $V$ as:	
\begin{equation}
	\label{eq:free_energy}
	F_{AL}(V,\beta) = -\lim_{N\to \infty} \frac{1}{N} \ln Z^{AL}_N(V,\beta)\,,
\end{equation}
where the partition function $Z^{AL}_N(V,\beta)$ is defined in \eqref{eq:partition_AL}. 
The next proposition shows that the free energy $F_{AL}(V,\beta)$ 
of the Generalized Gibbs ensemble of the Ablowitz--Ladik lattice and the free energy  $F_{HT}(V,\beta)$ in \eqref{eq:free_eig}  of the Circular $\beta$  ensemble at high-temperature  are related. This fact  allows us  to calculate the moments of the mean density of states of the CMV matrix $E$  in \eqref{E}  and of the Lax matrix $\mathcal{E}$  in \eqref{eq:Lax_matrix}. 

\begin{proposition}
	\label{PROP:MOMENT_RELATION}
	The free energy $F_{AL}(V,\beta)$  in  \eqref{eq:free_energy} of the AL lattice and the free energy $F_{HT}(V,\beta)$  in \eqref{eq:free_eig} of the Circular $\beta$  ensemble at high-temperature are analytic with respect to $\beta>0$, and are related by 
	\begin{equation}
		\label{eq:free_energy_relation}
		\partial_\beta \left(\beta F_{HT}(V,\beta)\right) + \ln(2) = F_{AL}(V,\beta).
	\end{equation}	
	The moments of the  density of states  $\mu^\beta_{AL}$   of the Lax matrix $\cE$ in  \eqref{eq:Lax_matrix} endowed with the probability measure \eqref{eq:GGE} and the moments of the density of states $\mu_{HT}$   of the Circular $\beta$  ensemble in the high-temperature regime \eqref{eq:HT_alpha} are related to the free energies 
	$F_{AL}(V,\beta)$ and $F_{HT}(V,\beta)$  by
	\begin{equation}
		\label{eq:moments_free_energy}
		\begin{split}
			&\Re\int_\T e^{i\theta m }\mu_{AL}^\beta(\di \theta)= \partial_t  F_{AL}\left(V + \frac{t}{2}\Re(z^m),\beta\right)_{\vert_{t=0}}\,,\\
			&  \Re\int_\T e^{i\theta m}\mu_{HT}^\beta(\di \theta) = \partial_t  F_{HT}\left(V + \frac{t}{2}\Re(z^m),\beta\right)_{\vert_{t=0}},
		\end{split}
	\end{equation}
	and analogously for the imaginary part of the moments taking care of using the potential $V + \frac{t}{2}\Im(z^m)$.
\end{proposition}
Since the proof of this proposition is rather technical, we postpone it to Appendix \ref{appendix:lemma_transfer}. We are now ready  to prove the  first  main Theorem~\ref{THM:RELATION}.
\begin{proof}[Proof of  Theorem~\ref{THM:RELATION}.]
	
	First, we define $c_n(\beta) := \int e^{i\theta n} \mu_{AL}^\beta(\di \theta),$ 
	
	\noindent $d_n(\beta) := \int e^{i\theta n} \mu_{HT}^\beta(\di \theta)$. Since  the eigenvalues of $\cE$ lie on the unit circle, we deduce the following chain of inequalities:
	
	\begin{equation}
		\vert c_n(\beta)\vert  = \lim_{N\to\infty}\frac{\left \vert  \meanval{\mbox{Tr}(\cE^n)} \right \vert }{2N}\leq   \lim_{N\to\infty}\frac{\meanval{\left \vert  \mbox{Tr}(\cE^n) \right \vert }}{2N} \leq 1\, ,
	\end{equation}
	where the expectation in made according to the Gibbs measure.
	Thus, from Lemma \ref{lem:bay}, we obtain that the measure $\mu_{AL}^\beta(\di \theta)$ is uniquely characterized by its moments.  
	
	Next,  from Proposition~\ref{PROP:MOMENT_RELATION} and Remark \ref{rem:lip_moments} we obtain the relation
	\begin{equation}
		\label{eq:moment_relation_thm}
		c_n(\beta)=  \partial_\beta \left(\beta d_n(\beta)\right)\,  \;a.e
	\end{equation}
	between the moments of the measures $\mu^\beta_{AL}(\theta)$ and $\mu_{HT}^\beta(\theta)$ respectively.
	
	This, together with \eqref{eq:moment_relation_thm}  and Remark~\ref{rem:lip_moments} implies  
	\begin{equation}
		\mu^\beta_{AL}( \theta) = \partial_\beta \left(\beta \mu_{HT}^\beta(\theta)\right)\;\;a.e.\,.
	\end{equation}
\end{proof}
Our next main result provides an explicit expression of the mean density of states $\mu_{HT}(\theta)$ for  the potential  $V(z)=\eta\Re(z)$.
This generalizes the result by Gross and Witten  \cite{Gross_Witten} and Baik-Deift-Johansson \cite{BDJ} obtained  for finite temperature  to the high-temperature regime.

\section{Proof of Theorem \ref{THM:MEAN_DENSITY}}
\label{sec4}
The proof of Theorem~\ref{THM:MEAN_DENSITY} consists of mainly two  parts: we first derive from the variational equations with respect to the  functional $\cF^{(V,\beta)}$, the double confluent Heun equation 
\eqref{eq:heun}. Then we show that such equation admits an analytic solution in any compact sets of the complex plane containing the origin.
From Theorem \ref{thm:lambert} we know that the density $\mu^\beta_{HT}$ is characterized as the unique minimizer of the functional \eqref{eq:functional}.
We follow the  ideas developed in \cite{Allez2012,Allez2013,Forrester2021,DKM} to find this minimizer explicitly.
We consider the Euler-Lagrange equation of the functional  \eqref{eq:functional}, namely
\begin{equation}
	\label{eq:VarEq}
	\dfrac{\delta \cF^{(V,\beta)}}{\delta \mu}=2V(\theta) -2\beta \int_{\T} \ln\sin\left(\frac{\vert \theta-\phi
		\vert }{2}\right)\mu(\phi)\di \phi  +\ln \mu(\theta)=C(V,\beta)\, ,\;\;a.e.
\end{equation}
where  the equation holds almost everywhere, $C(V,\beta)$ is a constant depending on the potential and $\beta$, but not on the variable $\theta$. 
Differentiating the Euler-Lagrange equation \eqref{eq:VarEq} at the minimizer $\mu_{HT}^\beta(\theta)$  with respect to $\theta$ we obtain the following integral equation (see \cite[Proposition 2.5]{Hardy2020}):

\begin{equation}
	\label{eq:euler_der}
	\partial_\theta \mu_{HT}^\beta(\theta)+\mu_{HT}^\beta(\theta)[2\partial_\theta (V(\theta))+\beta \mathcal{H}\mu_{HT}^\beta(\theta)] =0\, ,
\end{equation}
where   $\mathcal{H}$ is the Hilbert transform  defined on $L^2(\T)$ as
\begin{equation}
	\label{eq:hilbert_ransform}
	\mathcal{H}\mu_{HT}^\beta(\theta)= -\mbox{p.v.}\int_{\T}\cot\left(\frac{\theta-\phi}{2}\right)\mu_{HT}^\beta(\phi)\di \phi
\end{equation}
and $\mbox{p.v.}$ is the Cauchy principal value, that is the limit as $\varepsilon\to 0$ of the integral  on the torus $\T$  restricted to the domain $\vert e^{i\theta}-e^{i\phi}\vert >\varepsilon$. We notice that the Hilbert transform $\mathcal{H}$ is diagonal on the bases of exponential $\{e^{in\theta}\}_{n\in \Z}$, meaning that

\begin{equation}
	\itH e^{in\theta} = 2\pi i\mbox{sgn}(n)e^{in\theta}\, ,
\end{equation}
where $\mbox{sgn}(\cdot)$ is the sign function with the convention that $\mbox{sgn}(0) = 0$. 

Setting $e^{i\theta}=z$ and $e^{i\phi}=w$, we recognize the Riesz--Herglotz  kernel  $\dfrac{z+w}{z-w}$ expressed as 
\[
\dfrac{z+w}{z-w}=-i\cot\left(\frac{\theta-\phi}{2}\right)\,.
\]
Therefore 
\[
\int_{\T}\cot\left(\frac{\theta-\phi}{2}\right)\mu(\phi)\di \phi= i + 2\int_{S^1}\dfrac{\mu(\phi)_{\vert e^{i\phi} = w}dw}{z-w},
\]
where $S^1$ is the anticlockwise oriented circle,  and  we used the normalization condition  $\int_{\T}\mu(\phi)\di \phi=1$. In  the following, in  order to simplify  the notation, we indicate  $\mu(\phi)_{\vert e^{i\phi} = w}$ just as $\mu(w)$.
We can recast \eqref{eq:euler_der} in the form
\begin{equation}
	\label{eq:var1}
	z \partial_z \mu(z)+\mu(z)\left[2z \partial_zV(z)-\beta+ 2i\beta\mbox{p.v.}\int_{S^1}\mu(w)\dfrac{\di w }{z-w}\right]=0\,.
\end{equation}
For $z\in\C\backslash S^1$ let us define
\begin{equation}
	\label{eq:F_def}
	G(z) := \int_{S^1}\mu(w)\dfrac{\di w }{w-z} = \frac{i}{2} -  \frac{1}{2}\int_{\T}\cot\left(\frac{\theta-\phi}{2}\right)\mu(\phi)\di \phi\,,
\end{equation}
and for $z\in S^1$  let $G_\pm(z) = \lim_{\wt z \to z} G(\wt z) $ for $\wt z$  inside and outside  the unit circle respectively.
Then  by \eqref{eq:var1}
\begin{equation}
	\begin{split}
		G_{\pm}(z)  & = \pm \pi i\mu(z) + \mbox{p.v.}\int_{S^1}\mu(w)\dfrac{\di w }{w-z} \\
		& =  \pm \pi i \mu(z)  + \frac{i}{2} - \frac{2iz \partial_z V(z) }{2\beta} - \frac{i z \partial_z \mu(z)}{2 \beta \mu(z)}\,.
	\end{split}
\end{equation}
This implies that for $z\in S^1$ one has 
\begin{align}
	&G_+(z) + G_-(z) = i - \frac{2iz \partial_z V(z) }{\beta} - \frac{i z \partial_z \mu(z)}{\beta \mu(z)} \, ,\\
	&G_+(z) - G_-(z) = 2  \pi i\mu(z)\,.
\end{align}
Multiplying the two previous expressions, one obtains:
\begin{equation}
	G_+(z)^2 -G_-(z)^2 = 2 \pi i \mu(z)\left( i - \frac{2iz\partial_zV(z) }{\beta} - \frac{ z i \partial_z \mu(z)}{\beta \mu(z)}\right)\,.
\end{equation}
In order to proceed we have to specify the potential $V(z)$, in our case we will consider 
\begin{equation}
	\label{BDJ}
	V(z) = \frac{\eta}{2}\left(z + \frac{1}{z}\right).
\end{equation}
Applying  the Sokhtoski-Plemelj  formula \cite{Gakhov1990} to the above boundary value problem, one obtains  \begin{equation}
	\label{eq:fquadro}
	G^2(z) = i\int_{S^1} \frac{\mu(w)}{w-z}\di w - \frac{i\eta}{\beta}\int_{S^1}\frac{\left(w -\wo w\right) \mu(w)}{w-z} \di w - \frac{i}{\beta}\int_{S^1}\frac{w \partial_w\mu(w)}{w-z} \di w\, .
\end{equation}
The second term in the r.h.s. of the above expression gives
\begin{equation}
	\begin{split}
		\int_{S^1}\frac{(w-\wo w) \mu(w)}{w-z} \di w &=  \int_{S^1} \frac{w \pm z }{w-z}\mu(w)\di w + \frac{1}{z}\int_{S^1}\mu(w)\left(-\frac{1}{w-z} + \frac{1}{w} \right)\di w \\ 
		&=\left( z G(z) + i\lambda -\frac{G(z)}{z} + \frac{i}{z}\right)\, ,
	\end{split}
\end{equation}
where we have defined 
\begin{equation}
	\label{lambda}
	\lambda:=-i \int_{S^1}\mu(w) \di w,\quad  \lambda \in \R.
\end{equation}
The third term in the r.h.s. of \eqref{eq:fquadro} gives
\begin{equation}
	\begin{split}
		\int_{S^1}\frac{w \partial_w\mu(w)}{w-z} \di w & = \int_{S^1}\partial_w\mu(w) \di w  + z \int_{S^1}\frac{ \partial_w\mu(w)}{w-z} \di w \\
		& =z\int_{S^1}\frac{\mu(w)}{(w-z)^2}\di w = z \partial_z G(z),
	\end{split}
\end{equation}
where in these last relations we use the results of Theorem~\ref{thm:lambert} about the regularity of $\mu$.
Now we can rewrite \eqref{eq:fquadro} as
\begin{equation}
	\label{eq:Fsquaredfinal}
	G^2(z) = i G(z) -\frac{i\eta}{\beta}\left( z G(z)  +i\lambda  - \frac{G(z)}{z} + \frac{i}{z}\right) -\frac{i z \partial_z G(z)}{\beta}\,.
\end{equation}
\begin{remark}
	\label{rem:lambda}
	In the above ODE, the parameter $\lambda = \lambda(\eta,\beta)$  depends  via \eqref{lambda} implicitly on the function $G(z)$.
	Our strategy to solve the above equation is to consider $\lambda $ as a free parameter that is uniquely fixed by the analytic properties of the function $G(z)$.
	
\end{remark}
We can now turn the non-linear first order ODE \eqref{eq:Fsquaredfinal} into a linear second order ODE through the  substitution 
\begin{equation}
	\label{Ftov}
	G(z) = i + \frac{i z v'(z)}{\beta v(z)},
\end{equation}
getting:
\begin{equation}
	\label{eq:heun2}
	z^2 v''(z) + \left(-\eta + z (\beta +1) + \eta z^2\right)v'(z) +\eta\beta (z +\lambda) v(z) =0\, ,
\end{equation}
which is the DCH equation in  \eqref{eq:heun}.
The solutions to this  equation have generically essential singularities  at $z=0$ and $z=\infty$   and the local description near the singularities  depends on the parameter $\eta$ and $\beta$.
Indeed we have that the two fundamental solutions near $z=0$    have the following asymptotic behaviour

\begin{align}
	\label{eqv_1}
	&v^{(0)}_1(z)=e^{\eta(z+\frac{1}{z})}z^{1-\beta}\kappa_1(\eta,\beta,\lambda;z),\quad  - \frac{3\pi}{2} < \arg(\eta z) < \frac{3 \pi}{2}\, , \\
	\label{eqv_2}
	&v^{(0)}_2(z)=\kappa_2(\eta,\beta,\lambda,z),\quad  - \frac{\pi}{2} < \arg(\eta z) < \frac{5 \pi}{2}\, ,
\end{align}
where $\kappa_{j}(\eta,\beta,\lambda;z)$, $j=1,2$,  are asymptotic series in a neighbourhood of $z=0$.
The  quantity $\lambda$ is usually referred to as \textit{accessory parameter}.
Since  $G(z)$ is analytic in the unit disk,  continuous up to the boundary, and $G(0)=i$,  we deduce that 
\[
v(z)=v_0\exp\left[-i\int_{0}^z \beta\frac{G(s)-i}{s}ds\right],\quad v_0\neq 0,
\]
has to be analytic in the unit disk.
For this reason   we  seek for a solution $v(z)$ of the DCH equation   that is  analytic  in the unit disk  and such that $v(z) \xrightarrow[z\to 0]{} v_0 $, where $v_0$ is a nonzero constant.  
\paragraph*{Construction of the analytic solution of  equation \eqref{eq:heun2}.}
Of the fundamental solutions \eqref{eqv_1} and \eqref{eqv_2} of equation \eqref{eq:heun2}   only the solution \eqref{eqv_2} has a chance of being analytic near $z=0$. This  occurs  if we are able to make the asymptotic series defined by $\kappa_2(\eta,\beta,\lambda,z)$, into a convergent series. We look for a solution   of \eqref{eq:heun2}  in the form {\color{black} of a convergent power series}
\begin{equation}
	\label{eq:taylor_series}
	v(z) = \sum_{k=0}^\infty a_k z^k,
\end{equation}
where $a_k=a_k(\eta,\beta,\lambda)$.
This  implies the following recurrence relations for the coefficients  $\{a_k\}_{k\in \N}$
\begin{align}
	\label{eq:zero_recurrence}
	&\eta(a_0\lambda\beta - a_1) =0\, , \\
	\label{eq:recurrence_taylor}
	&a_k(k^2 + k\beta  + \lambda\beta \eta ) + \eta (k-1+\beta) a_{k-1} - \eta(k+1) a_{k+1} = 0 \,, \quad k > 0\, ,
\end{align} 
where we have the freedom to chose $\lambda$ and $a_0$. Generically, the above recurrence relation for the coefficients $\{a_k\}_{k\in \N}$
gives a divergent series in \eqref{eq:taylor_series}.  To obtain a convergent series, we follow the ideas in  \cite{Tertychniy2007,Buchstaber2015}.

We start by considering the $2\times 2 $ matrices $R^{(s)}_k$ defined as 
\begin{equation}
	\label{Rks}
	R_k^{(s)}= M_kM_{k+1}\dots M_s,\; s\geq k,\;\;  M_k = \begin{pmatrix}
		1 + \frac{\lambda\beta \eta}{k(k+\beta)} & \frac{\eta^2}{k(k+\beta+1)}\\
		1 & 0
	\end{pmatrix}\,,
\end{equation}
which satisfy  the recurrence relation  $R_k^{(s)}=R_k^{(s-1)}M_s$.  The next lemma shows that the limit of  $R_k^{(s)}$  as $s\to\infty$  exists. 
\begin{lemma}
	\label{lemma_R}
	Let $R_k^{(s)}$  be the matrix defined in  \eqref{Rks}. Then the  limit  of   $R_k^{(s)}$  as $s\to\infty$ exists and 
	\begin{equation}
		\label{eq:R_def}
		R_k:=\lim_{s\to\infty}R_k^{(s)}.
	\end{equation}
	The matrices $R_k$, $k\geq 1$  satisfy  the descending recurrence relation:
	\begin{equation}
		\label{eq:R_recurrence}
		R_k = M_k R_{k+1}\quad k \geq 1\, .
	\end{equation} 
	Furthermore each entry of the matrix $R_k=R_k(\beta,\eta,\lambda)$ is differentiable with respect to the parameters $\beta, \, \eta$, and $\lambda$.
\end{lemma}
Since the proof of this lemma is rather technical, we defer it to appendix \ref{app:B}.
Further, let us define the following function:
\begin{equation}
	\label{eq:little_xi_def}
	\xi(\eta,\beta,\lambda) :=  \begin{pmatrix}
		\lambda & \frac{\eta}{\beta +1} 	\end{pmatrix} R_1 \begin{pmatrix}
		1 \\ 0
	\end{pmatrix} .
\end{equation}
We are now ready to prove the following result that will give us 
a necessary condition to fix the value of $\lambda$.
{
	\begin{proposition}
		\label{thm:analytic_confluent}
		For the values of $\lambda$ such that 
		\begin{equation}
			\label{eq:zero_condition}
			\xi(\eta,\beta,\lambda) =0,
		\end{equation}
		where  $ \xi(\eta,\beta,\lambda) $ is defined in \eqref{eq:little_xi_def}, the  Double Confluent Heun equation \eqref{eq:heun2} admits a non-zero   solution   $v=v(z,\eta,\beta)$   defined by the series 
		\eqref{eq:taylor_series} that is uniformly convergent in $\vert z\vert \leq r$ with $r\geq 1$. The   corresponding coefficients $\{a_k\}_{k\in \N}$ of the Taylor expansion \eqref{eq:taylor_series}
		are given by the relation
		\begin{align}
			\label{eq:reconstruction_0}
			&a_0 = \frac{1}{\beta}\begin{pmatrix}
				1 & 0 
			\end{pmatrix} R_1 \begin{pmatrix}
				1 \\ 0
			\end{pmatrix}\, , \\
			\label{eq:reconstruction}
			& a_k = (-1)^k\frac{\eta^k}{k!(k+\beta)} \begin{pmatrix}
				0 & 1
			\end{pmatrix}  R_k \begin{pmatrix}
				1 \\ 0 
			\end{pmatrix}\, , \quad k \geq 1 \, \, , 
		\end{align}
		where the matrices $R_k$ are defined in \eqref{eq:R_def}.	For each $\lambda$ satisfying  \eqref{eq:zero_condition}, the solution  $v(z)$ of the DCH equation  \eqref{eq:heun2},  analytic at zero is unique up to a multiplicative  factor.
	\end{proposition}
}

\begin{proof}
	First, we show that choosing  $a_k$ according to \eqref{eq:reconstruction_0}-\eqref{eq:reconstruction} we obtain  a solution of the recurrence \eqref{eq:recurrence_taylor}. 
	We notice that due to the recurrence relation for the matrices $R_k$ \eqref{eq:R_recurrence}, we have that:	 
	\begin{equation}
		\begin{pmatrix}
			0& 1
		\end{pmatrix} R_k \begin{pmatrix}
			1 \\ 0
		\end{pmatrix} = \begin{pmatrix}
			1 & 0 
		\end{pmatrix} R_{k+1 } \begin{pmatrix}
			1 \\ 0
		\end{pmatrix}\, .
	\end{equation}
	Thus, applying the previous equation and \eqref{eq:reconstruction_0}-\eqref{eq:reconstruction}, we can recast \eqref{eq:recurrence_taylor} as:
	\begin{equation}
		\begin{split}
			&\left[ (-1)^{k-1} \frac{\eta^k}{(k-1)!} \begin{pmatrix}
				1 & 0
			\end{pmatrix} R_k + (-\eta)^k\ \frac{k(k+\beta) + \eta\lambda\beta}{k!(k+\beta)}   \begin{pmatrix}
				1 & 0
			\end{pmatrix} R_{k+1} \right.\\
			&\left.+ (-1)^{k} \frac{\eta^{k+2}}{k!(k+1+\beta)}  \begin{pmatrix}
				0 & 1
			\end{pmatrix} R_{k+1}\right] \begin{pmatrix}
				1 \\ 0
			\end{pmatrix} \\ 
			& = \frac{(-\eta)^{k}}{(k-1)!}\left[ - \begin{pmatrix}
				1 & 0 
			\end{pmatrix}R_k + \begin{pmatrix}
				1 + \frac{\lambda \beta\eta}{k(k+\beta)} & \frac{\eta^2}{k(k+1+\beta)}
			\end{pmatrix} R_{k+1}\right] \begin{pmatrix}
				1 \\ 0
			\end{pmatrix}  = 0\, ,
		\end{split}
	\end{equation}
	where in the last equality we have enforced \eqref{eq:R_recurrence}.
	Next  we can rewrite \eqref{eq:zero_recurrence} in terms of the matrix $R_1$ exploiting \eqref{eq:reconstruction_0}-\eqref{eq:reconstruction}, namely
	\begin{equation}
		0= \begin{pmatrix}
			\lambda & \frac{\eta}{\beta +1} 	\end{pmatrix} R_1 \begin{pmatrix}
			1 \\ 0
		\end{pmatrix}  = \xi(\eta,\beta,\lambda)\, ,
	\end{equation}
	which  is exactly \eqref{eq:little_xi_def}.
	Since the entries of the matrices $R_k$ are uniformly bounded,  the solution $v(z)=\sum_{k\geq 0}a_kz^k$  with $a_k$ as in \eqref{eq:reconstruction},  defines a uniformly convergent Taylor series  in $\vert z\vert <r$ for any $r\geq 0$ and in particular for any $r>1$. 
	
	To show that the solution  analytic at $z=0$ is unique  up to a constant,  we consider the Wronkstian $W(v,\tilde{v})(z)$ 
	of two independent solution   $v$ and $\tilde{v}$  of the  Double Confluent Heun equation
	\eqref{eq:heun2}, namely   
	\[
	W(v,\wt v)(z)=e^{-\eta(z+\frac{1}{z})}z^{-(\beta+1)}(v'(z)\tilde{v}(z)-v(z)\tilde{v}'(z)).
	\]
	Since $W'(v,\wt v)(z)=0$, it follows that $W(v,\wt v)(z)=C$ a constant. If by contradiction we suppose that there are two analytic solutions at $z=0$, then from the above relation we obtain
	\[
	e^{-\eta z}(v'(z)\tilde{v}(z)-v(z)\tilde{v}'(z))=Ce^{\frac{\eta}{z}}z^{\beta+1}\,.
	\]
	When $\eta \ne 0$ the left-hand side of the above  equation is analytic and the right-hand side is not, that is clearly a contradiction. When  $\eta=0$  then \eqref{eq:heun2} becomes:
	\begin{equation}
		z^2 v''(z) + z(\beta+1)v'(z) =0\, .
	\end{equation}
	The above   equation has two independent solutions, one is the constant solution, which is analytic, the other one is  $v(z) = Cz^{-\beta}$ which is not analytic since $\beta>0$.
	
\end{proof}

\begin{remark}
	We observe that the equation  \eqref{eq:zero_condition}  does not uniquely determine $\lambda$.  Indeed, as it is shown in Figure \ref{fig:non_unique} the function $\xi(\eta,\beta,\lambda)$ may have several zeros for  given  $\eta$ and $\beta$.
\end{remark}
\begin{figure}[ht]
	\centering
	\includegraphics[scale=0.2]{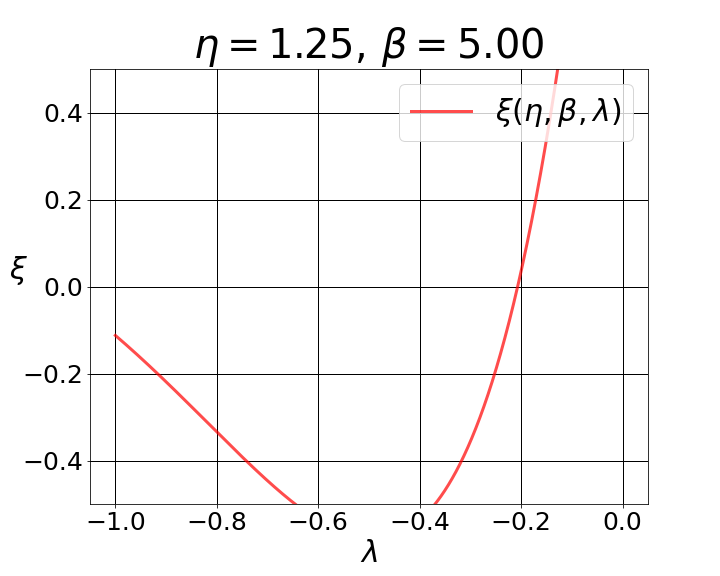}
	\includegraphics[scale=0.2]{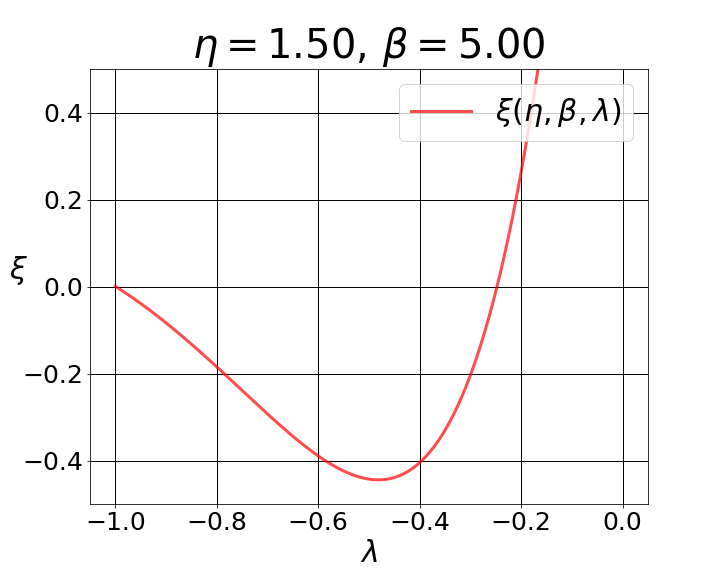}
	
	\includegraphics[scale=0.2]{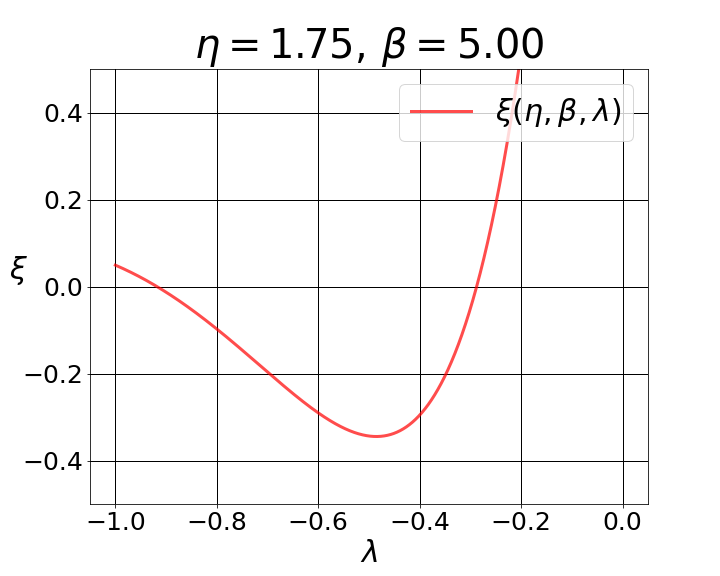}
	\includegraphics[scale=0.2]{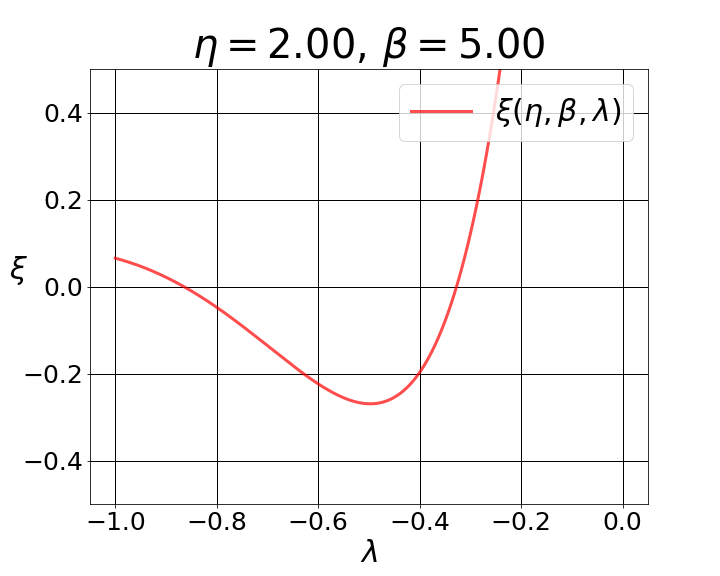}
	
	\caption{Plots of $\xi(\eta,\beta,\lambda)$ for various values of $\eta,\beta$}
	\label{fig:non_unique}
\end{figure}

\paragraph*{Choice of the parameter $\lambda$.}
We will now prove that  the parameter $\lambda$ is uniquely determined in  a neighbourhood of $\eta=0$ by requiring that the solution $v=v(z,\eta,\beta)$ 
depends continuously on the parameter $\eta$.

\begin{lemma}
	\label{LEM:UNIQUE}
	There exists an $\varepsilon >0$ such that for all $\eta \in (-\varepsilon,\varepsilon)$ and $\beta > 0 $ there is a unique $\lambda = \lambda(\eta,\beta)$ such that $\xi(\eta, \beta, \lambda(\eta,\beta)) = 0$.
\end{lemma}

\begin{proof}
	When $\eta=0$   the matrix $R_j=\begin{pmatrix}1&0\\1&0\end{pmatrix}$ so  that the only solution of the equation 
	\eqref{eq:zero_condition}   $\xi(\eta=0,\beta, \lambda) = 0$ is $\lambda = 0$.
	To show the existence of the solution \eqref{eq:zero_condition} for $\lambda=\lambda(\eta,\beta)$  near $\eta=0$, we use the implicit function theorem.
	We have to show that $\partial_\lambda \xi(\eta,\beta,\lambda)_{\vert_{(0,\beta,0)}} \ne 0$.
	For this purpose, we need to  evaluate
	\[ \partial_\lambda \left(M_k\right)_{(\eta= 0,\lambda = 0)} = \begin{pmatrix}
		0 & 0 \\
		0 & 0
	\end{pmatrix}\, , \]	
	where $M_k$ is defined in \eqref{eq:R_def}. This equation implies that \[
	\partial_\lambda(\xi(\eta,\beta,\lambda))_{\vert_{(0,\beta,0)}} = \begin{pmatrix}
		1& 0\end{pmatrix} \begin{pmatrix}
		1 & 0 \\
		1 & 0
	\end{pmatrix} \begin{pmatrix}
		1 \\ 0
	\end{pmatrix}=1.
	\]
	Thus we can apply the implicit function theorem, and we obtain the claim.
\end{proof}
\color{black}

We conclude the proof of Theorem~\ref{THM:MEAN_DENSITY}. 
When $\eta=0$ the only analytic  solution of DCH equation is  $v(z)=c$, $c\in \C\backslash\{0\}$.  In this case, in principle $\lambda$ is undetermined. However, from Theorem~\ref{thm:lambert}
the minimizer  $\mu^\beta_{HT}$ of \eqref{eq:functional} is the uniform measure on the circle and therefore from equation \ref{lambda} one has $\lambda=0$. 
From Lemma~\ref{LEM:UNIQUE}  when  $\eta\in(-\varepsilon,\varepsilon)$, there exists a unique $\lambda(\eta,\beta)$  that satisfies equation \eqref{eq:zero_condition}  and such that $\lambda(\eta=0,\beta)=0$
and therefore  by Proposition~\ref{thm:analytic_confluent} we obtain for $\eta\in(-\varepsilon,\varepsilon)$, the unique  solution  $v(z,\eta,\beta)$ of the DCH equation analytic in any compact set $\vert z\vert \leq r$, with $r>0$ and in particular  when $r=1$.  
Because of lemma~\ref{lemma_R} the solution $v(z,\eta,\beta)$ is differentiable with respect to the parameters $\eta$ and $\beta$.

We remark that $v(z)\neq 0$ on the unit disc $\overline{\D}$ because of the relation  
\eqref{Ftov} between the analytic function $G(z)$ and $v(z)$  and the uniqueness  of the minimizer $\mu^{\beta}_{HT}$ and of the analytic solution $v(z)$ of  \eqref{eq:heun2}.

To complete our proof   of Theorem~\ref{THM:MEAN_DENSITY} we  recover the explicit expression of $\mu^\beta_{HT}(\theta)$ from $G(z)$ and $v(z)$  using the  {\em Poisson representation formula} (see for example \cite[Chapter 1]{Simon2005}): 
\begin{equation}
	\label{mu_temp}
	\mu^\beta_{HT}(\theta) = -\frac{1}{2\pi} - \frac{\Re(i G(e^{i\theta}))}{\pi \beta}  = \frac{1}{2\pi} + \frac{1}{\pi \beta} \Re\left.\left(\frac{zv'(z)}{v(z)}\right\vert _{z=e^{i\theta}}\right)\, .
\end{equation}  
\qed

\begin{figure}[ht]
	\centering
	\includegraphics[scale = 0.2]{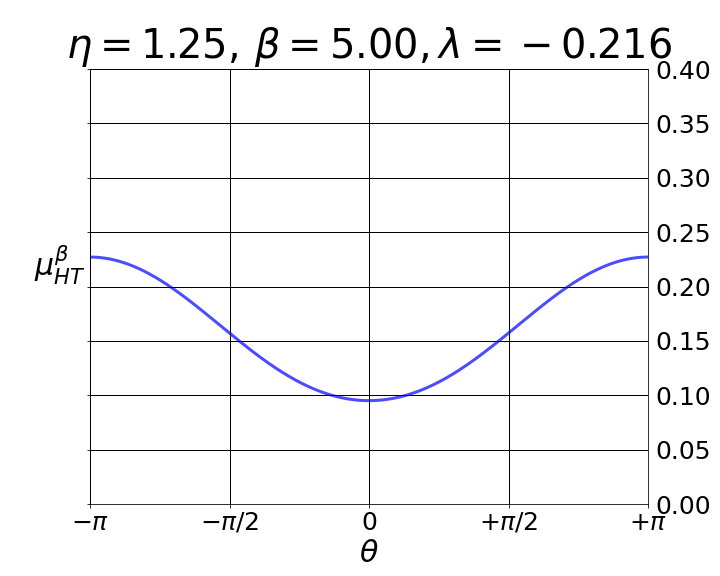}
	\includegraphics[scale = 0.2]{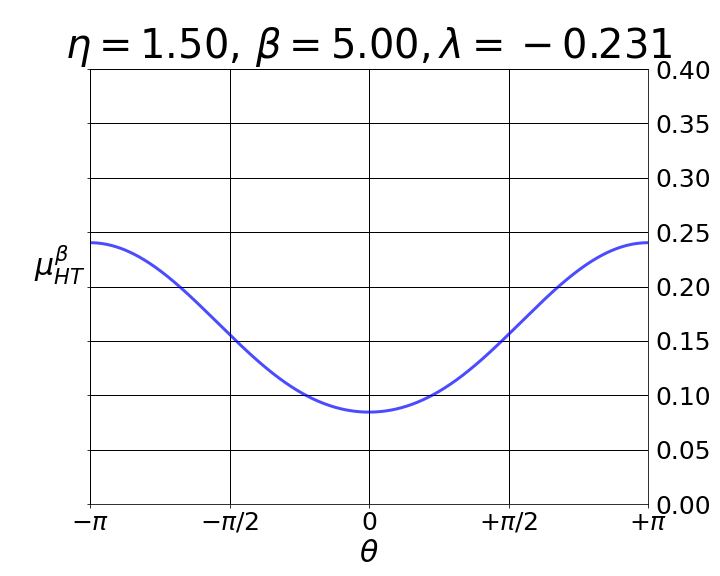}

	\caption{The mean density of states $\mu^\beta_{HT}$ for different parameters.}
	\label{fig:density}
\end{figure}

In Figure \ref{fig:density} we plotted the density of states of the Circular $\beta$  ensemble in the high-temperature regime with potential $V(z) = \eta\Re(z)$. To produce this picture and Figure \ref{fig:non_unique}, we used extensively the \texttt{NumPy} \cite{numpy}, and \texttt{matplotlib} \cite{matplotlib} libraries.
\begin{remark}
	\label{rem_Deift}
	The Gross-Witten  \cite{Gross_Witten} and Baik-Deift-Johannson \cite{BDJ} solution 
	is obtained by 
	making the substitution  $\eta\to\beta\eta$ and $\beta\to\infty$  in equation
	\eqref{eq:functional}  which gives the functional 
	\[
	\cF^{(\eta)}(\mu):=  \int\int\limits_{\T\times \T} \ln\left\vert \sin\left(\frac{\theta-\phi}{2}\right)\right\vert ^{-1}\mu(\di \theta)\mu(\di \phi )+2\eta\int_{\T}\cos(\theta)\mu(\di \theta).
	\]  
	The minimizer is $\mu(\theta)= \frac{1}{2\pi}\left(1 - 2\eta \cos\theta\right)$  with $0\leq 2\eta\leq 1$. In this case the first moment $\lambda=-\eta$. 
	
\end{remark}
\begin{appendices}
	
	\section{Proof of Proposition \ref{PROP:MOMENT_RELATION}}
	\label{appendix:lemma_transfer}
	
	First, we prove the relation between the free energies \eqref{eq:free_energy_relation}  namely:
	
	\begin{equation}
		\label{eq:rel_free}
		\partial_\beta(\beta F_{HT}(V,\beta)) + \ln(2)= F_{AL}(V,\beta)\, ,
	\end{equation}
	and we show that  $ F_{HT}(V,\beta)$ is   analytic with respect to $\beta>0$.
	
	From Remark~\ref{rem:free_energy}, the above expression is equivalent to:
	
	\begin{equation}
		\partial_\beta\left(\beta  \lim_{N\to \infty} \frac{\ln(Z^{HT}_N(V,\beta))}{N}\right) =  \lim_{N\to \infty} \frac{\ln(Z^{AL}_N(V,\beta))}{N}\,.
	\end{equation}
	To prove this  relation, we will use the so-called transfer operator technique \cite{Kevrekidis2001,Krumhansl1975,Peyrard1989}. We are considering a potential of the form  $\mbox{Tr}(V(\cE))$ as in \eqref{eq:potential} which is of finite range $K$, meaning that it can be expressed as a sum of local quantities, i.e. depending on a finite number $2K$ of variables, with $K$  independent of  $N$  \cite{Nenciu2005}. For example, if $V(z) =\Re( z)$, then $\mbox{Tr}(\cE) = -2\sum_{j=1}^N\Re( \alpha_j\wo \alpha_{j+1})$ and in this case the range  is $K=1$. Let $N = KM + L$ with  $M,L\in \N$ and $L< K$.  We split the coordinates $(\alpha_1,\dots,\alpha_N)$  into  $M$ blocks  of length $K$ and a reminder of length $L$,  and we define  the vector $\wt \balpha_j $  of length  $K$ as $$\wt \balpha_j = (\alpha_{K(j-1)+1}, \alpha_{K(j-1)+2}, \ldots, \alpha_{Kj}).$$  In this notation,
	\[
	(\alpha_1,\dots,\alpha_N)=(\overbrace{\wt \balpha_1,\dots,\wt \balpha_M}^{KM},\overbrace{\alpha_{KM+1}, \ldots,\alpha_N}^{L}),
	\]
	\begin{equation}
		\begin{split}
			\mbox{Tr}(V(\cE))& = \sum_{\ell=1}^{M-1} W(\wt \balpha_\ell, \wt \balpha_{\ell+1}) + W(\wt \balpha_{M}, \overbrace{\alpha_{KM+1}, \ldots,\alpha_{KM+L}}^{L}, \overbrace{\alpha_1,\ldots, \alpha_{K-L}}^{K-L})\\
			&+W_1(\overbrace{\alpha_{KM+1}, \ldots,\alpha_{KM+L}}^{L}, \wt\balpha_1),
		\end{split}
	\end{equation}
	where $W \,:\, \overline{\D}^{K}\times \overline{\D}^{K} \to \R $ and $W_1 \,:\, \overline{\D}^{L}\times \overline{\D}^{K} \to \R $  are  continuous functions.   The last  two  terms  in the above expression are  different from the others since we may have an off-set of length $L$, due to periodicity.   	In the case $V(z) = \Re(z)$, then $W(\alpha_1,\alpha_2) = -2\Re(\alpha_1\wo \alpha_2)$,  
	there is no off-set and $W_1=0$.
	
	For convenience, we define
	\[
	\wt \balpha_{M+1}=(\alpha_{KM+1}, \ldots,\alpha_{KM+L}, \alpha_1,\ldots, \alpha_{K-L}).
	\]
	We can now rewrite $Z^{AL}_N(V,\beta)$  in \eqref{eq:partition_AL} as
	\begin{equation}
		\label{eq:partition_1}
		\begin{split}
			Z^{AL}_N(V,\beta) &= \int\limits_{\D^N} \prod_{j=1}^{N} \left(1-\vert \alpha_j\vert ^2\right)^{\beta-1}\\ &\times \exp\left(- \sum_{\ell=1}^{M} W(\wt \balpha_\ell, \wt \balpha_{\ell+1}){ -W_1(\alpha_{KM+1}, \ldots,\alpha_{KM+L}, \wt\balpha_1)}\right)  \di^2 \balpha \,.
		\end{split}
	\end{equation}
	We are now in position to apply the transfer operator technique to compute this partition function. On $L^2(\D^K)$ we introduce the scalar product 	
	
	\begin{equation}
		\label{eq:scalar_prod}
		\left( f,g\right) = \int_{\D^K} f(\bz) \wo{g(\bz)} \di \bz\,,
	\end{equation}
	where $\bz = (z_1, \ldots, z_K)$. This scalar product induces a norm on $L^2(\D)$ and also a norm on bounded  operators $T \;: \, L^2(\D^K) \to L^2(\D^K)$ as
	
	\begin{equation}
		\vert \vert T\vert \vert  := \sup_{f\,:\, \vert \vert f\vert \vert _2 = 1} \vert \vert Tf\vert \vert _2\, ,
	\end{equation}
	where $\vert \vert f\vert \vert _2$ is the standard $L^2$ norm.
	
	Let  $\bzeta=(\zeta_1,\ldots,\ldots\zeta_{2K})$ with $\zeta_{K+j}=\zeta_j>0$ for   $j=1,\dots, K$. We define  the  continuous  family of transfer operators $\cT_{\bzeta} \,:\, L^2(\D^K) \to L^2(\D^K)$ as
	
	\begin{equation}
		\label{eq:transfer}
		(\cT_{\bzeta}f)(\wt \balpha_2) = \int_{\D^K} f(\wt \balpha_1) \prod_{j=1}^{2K} \left(1-\vert \alpha_j\vert ^2\right)^{\frac{\zeta_j-1}{2}}\exp\left(- W(\wt \balpha_1, \wt \balpha_{2})\right)  \di^2\wt \balpha_1\, .
	\end{equation}
	We observe that  $\cT_{\bzeta}$  is an integral operator whose kernel $\prod_{j=1}^{2K} \left(1-\vert \alpha_j\vert ^2\right)^{\frac{\zeta_j-1}{2}}\exp\left(- W(\wt \balpha_1, \wt \balpha_{2})\right) $ belongs to $L^2(\D^K\times \D^K)$, and therefore $\cT_{\bzeta}$ is an Hilbert-Schimdt operator.
	We conclude that  there exists a complete set of   normalized eigenfunctions $\{\psi_j\}_{j\geq 1}$ with eigenvalues $\{\lambda_j\}_{j\geq 1}$  numbered so  that  $\{\vert \lambda_j\vert \}_{j\geq 1}$ is a  non-increasing sequence  such that:	
	\begin{align}
		\label{eq:eigenvalues_reduction}
		& (\cT_{\bzeta}\psi_j)(\bz,V,\bzeta) = \lambda_j(V,\bzeta) \psi_j(\bz,V,\bzeta)\, ,\\
		\label{eq:delta_rep}
		&\sum_{n=1}^{\infty} \wo \psi_n(\bz,V,\bzeta)\psi_n(\bz',V,\bzeta) = \delta_\bz(\bz')\,, 
	\end{align}  	
	where $\delta_\bz(\cdot)$ is the Dirac delta function at $\bz\in \D^K$.  
	
	For clearness, we collect a series of properties that the operator $\cT_{\bzeta}$ fulfils:
	
	\begin{itemize}
		\item[a)] $\sum_{j=1}^{\infty}\vert  \lambda_j(V,\bzeta)\vert ^2<\infty$ and $\cT_{\bzeta}$ is compact, since it  is Hilbert-Schimdt (see \cite[Chapter V.2.4]{KatoBook});
		\item[b)] {The eigenvalue $\lambda_1(\bzeta,V)$ is simple, positive and $\lambda_1(\bzeta,V) > \vert \lambda_n(\bzeta,V)\vert $ for all $n\geq 2$ (see \cite[Theorem 137.4]{ZaanenBook});}
		\item[c)] The  eigenvalue $\lambda_1(\bzeta,V)$ and its eigenfunction $\psi_1(\bz, \bzeta,V)$ are analytic functions of the parameters $\bzeta$, and for  any real polynomial $P$ there exists an $\varepsilon>0$ such that the maps $t\to\lambda_1(\bzeta,V+tP)$, $t\to\psi_1(\bz,\bzeta,V+tP)$ are analytic for $ \vert t \vert <\varepsilon$ (see \cite[Chapter VII, Theorem 1.8]{KatoBook}).
	\end{itemize}

	We artificially rewrite $Z^{AL}_N(V,\beta)$  in \eqref{eq:partition_1} as
	
	\begin{equation}
		\label{eq:arificial}
		\begin{split}
			&Z^{AL}_N(V,\beta)  =  \int\limits_{\D^{N+K}} \delta_{\wt \balpha_1}(\bgamma)\prod_{\ell = 1}^K\left[\left(1-\vert \gamma_\ell\vert ^2\right)\left(1-\vert \alpha_\ell\vert ^2\right)\right]^{\frac{\beta-1}{2}} \prod_{\ell = K+1}^N\left(1-\vert \alpha_\ell\vert ^2\right)^{\beta-1} \\ & \times \exp\left(- \sum_{\ell=1}^{M-1} W(\wt \balpha_\ell, \wt \balpha_{\ell+1}) - W(\wt \balpha_{M}, \alpha_{KM+1}, \ldots,\alpha_N, \gamma_1,\ldots, \gamma_{K-L})\right)\\
			&\times\exp\left(-W_1(\alpha_{KM+1}, \ldots,\alpha_{KM+L}, \bgamma)\right) \prod_{j=1}^N\di^2\alpha_j \di^2 \bgamma\, ,
		\end{split}
	\end{equation}
	where $\bgamma = (\gamma_1, \ldots, \gamma_K)$ and $\bgamma \in \D^K$. \\
	We can use \eqref{eq:delta_rep} with $\bzeta=\boldsymbol{\beta}=\overbrace{(\beta,\dots,\beta)}^{2K}$ to rewrite the previous equation as:
	\begin{equation}
		\begin{split}
			&Z^{AL}_N(V,\beta)  =  \int\limits_{\D^{N+K}}\sum_{n=1}^\infty\wo \psi_n(\bgamma,V,\boldsymbol{\beta})\psi_n(\wt\balpha_1,V,\boldsymbol{\beta})\\ &\times \prod_{\ell = 1}^K\left[\left(1-\vert \gamma_\ell\vert ^2\right) \left(1-\vert \alpha_\ell\vert ^2\right)\right]^{\frac{\beta-1}{2}} \prod_{\ell = K+1}^N\left(1-\vert \alpha_\ell\vert ^2\right)^{\beta-1} \\ & \times \exp\left(- \sum_{\ell=1}^{M-1} W(\wt \balpha_\ell, \wt \balpha_{\ell+1}) - W(\wt \balpha_{M}, \alpha_{KM+1}, \ldots,\alpha_N, \gamma_1,\ldots, \gamma_{K-L})\right) \\
			&\times\exp\left(-W_1(\alpha_{KM+1}, \ldots,\alpha_{KM+L}, \bgamma)\right) \di^2\balpha \di^2 \bgamma\\
			&=\sum_{n=1}^\infty \int_{\D^{N}}\wo \psi_n(\bgamma,V,\boldsymbol{\beta})
			(\cT_{\boldsymbol{\beta}}\psi_n)(\wt \balpha_2)
			\prod_{\ell=K+1}^{2K} \left(1-\vert \alpha_\ell\vert ^2\right)^{\frac{\beta-1}{2}} \prod_{\ell = 2K+1}^N\left(1-\vert \alpha_\ell\vert ^2\right)^{\beta-1} \\ & \times \exp\left(- \sum_{\ell=2}^{M-1} W(\wt \balpha_\ell, \wt \balpha_{\ell+1}) - W(\wt \balpha_{M}, \alpha_{KM+1}, \ldots,\alpha_N, \gamma_1,\ldots, \gamma_{K-L})\right)\\
			&\times\exp\left(-W_1(\alpha_{KM+1}, \ldots,\alpha_{KM+L}, \bgamma)\right)\prod_{\ell = 1}^K\di^2 \gamma_\ell \left(1-\vert \gamma_\ell\vert ^2\right)^{\frac{\beta -1}{2}}\prod_{\ell=K+1}^N\di^2 \alpha_\ell \\
			&=\sum_{n=1}^\infty \lambda_n(V,\boldsymbol{\beta})\int_{\D^{N}}\wo \psi_n(\bgamma,V,\boldsymbol{\beta})\psi_n(\wt \balpha_2,V,\boldsymbol{\beta})
			\prod_{\ell=K+1}^{2K} \left(1-\vert \alpha_\ell\vert ^2\right)^{\frac{\beta-1}{2}} \\ 
			&\prod_{\ell = 2K+1}^N\left(1-\vert \alpha_\ell\vert ^2\right)^{\beta-1} \exp\left(- \sum_{\ell=2}^{M-1} W(\wt \balpha_\ell, \wt \balpha_{\ell+1}) \right)\\
			&\times\exp\left(- W(\wt \balpha_{M}, \alpha_{KM+1}, \ldots,\alpha_N, \gamma_1,\ldots, \gamma_{K-L})\right) \\ &\times \exp\left(-W_1(\alpha_{KM+1}, \ldots,\alpha_{KM+L}, \bgamma)\right)\prod_{\ell = 1}^K \di^2 \gamma_\ell \left(1-\vert \gamma_\ell\vert ^2\right)^{\frac{\beta -1}{2}}\prod_{\ell=K+1}^N \di^2 \alpha_\ell\, . 
		\end{split}
	\end{equation}
	In the above integral, from the first to the second relation we 
	identify  the integral operator   $\cT_{\boldsymbol{\beta}}$  where $\boldsymbol{\beta}=\overbrace{(\beta,\dots,\beta)}^{2K}$.
	We  repeatedly apply  $\cT_{\boldsymbol{\beta}}$ and \eqref{eq:eigenvalues_reduction} another $M-2$ times      to the above integral,  to obtain:
	
	\begin{align}
		\label{eq:eigenvalues}
		&Z^{AL}_N(V,\beta)  = \sum_{n=1}^\infty (\lambda_n(V,\boldsymbol{\beta}))^{M-1}R_n,\\
		\label{Rn}
		&R_n=\int_{\D^{2K+L}} \overline{\psi}_n(\bgamma,V,\boldsymbol{\beta})\psi_n(\wt\balpha_M,V,\boldsymbol{\beta})
		\prod_{\ell = 1}^K \di^2 \gamma_\ell\left(1-\vert \gamma_\ell\vert ^2\right)^{\frac{\beta -1}{2}}\\
		\nonumber
		& \prod_{\ell=(M-1)K+1}^{MK} \left(1-\vert \alpha_\ell\vert ^2\right)^{\frac{\beta-1}{2}}\prod_{\ell =MK+1}^N\left(1-\vert \alpha_\ell\vert ^2\right)^{\beta-1}\\ 
		\nonumber
		&\times    \exp\left( - W(\wt \balpha_{M}, \alpha_{KM+1}, \ldots,\alpha_N, \gamma_1,\ldots, \gamma_{K-L})\right)\\
		\nonumber
		&\times\exp\left(-W_1(\alpha_{KM+1}, \ldots,\alpha_{KM+L}, \bgamma)\right)\prod_{\ell=(M-1)K+1}^N \di^2 \alpha_\ell.
	\end{align}
	
	The modulus of the reminder $\vert R_n\vert $  in \eqref{Rn} can be easily  bounded  from above and below by two constants $C_1,C_2 >0$ independent of $N$, therefore we conclude  from \eqref{eq:eigenvalues} that 
	
	\begin{equation}
		\label{eq:FAL}
		F_{AL}(V,\beta) = -\lim_{N\to \infty}\frac{1}{N} \ln\left(Z^{AL}_N(V,\beta) \right) = -\frac{1}{K}\ln\left(\lambda_{1}(V,\boldsymbol{\beta})\right)\, .
	\end{equation}	
	Since $\lambda_1(V,\boldsymbol{\beta})$ is analytic for $\beta>0$, see \cite[Chapter VII, Theorem 1.8]{KatoBook}, and strictly positive, see \cite[Theorem 137.4]{ZaanenBook}, we conclude that $F_{AL}(V,\beta)$ is analytic with respect to $\beta$.

	We can apply  the same procedure to the partition function $ Z^{HT}_N(V,\beta)$  in \eqref{eq:HT_alpha}. Also in this case the potential
	$\mbox{Tr}(V(\wt E))$ with $V$ as in \eqref{eq:potential} and the matrix $\wt E$ as in \eqref{eq:doubble_E}  is of finite range $K$, meaning that it can be expressed as a sum of local quantities  \cite{Nenciu2005}.   
	More precisely, assuming $N=KM+L$ with $L<K$ and $M,N,L\in\N$ we have
	\begin{equation}
		\label{eq:potentialE}
		\begin{split}
			\mbox{Tr}(V(\wt E))& = \sum_{\ell=1}^{M-1} W(\wt \balpha_\ell, \wt \balpha_{\ell+1}) + W(\overbrace{0,\ldots,0}^{K-1}, -1 ,\wt \balpha_1) \\
			&+W(\wt \balpha_{M}, \underbrace{\alpha_{KM+1}, \ldots,\alpha_N}_{L},\underbrace{0,\dots,0}_{K-L})\, .
		\end{split}
	\end{equation} 
	For example for $V(z)=z^2+\bar{z}^2$ one has $K=2$ and $N=2M+L$ where $L=0,1$, depending on the parity of $N$.
	The vector  $\wt \balpha_\ell $ takes the form 	 $\wt \balpha_\ell = (\alpha_{2\ell-1},\alpha_{2\ell})$ for $\ell=1,\dots, M$.  In this notation,
	we can rewrite the potential as
	\begin{equation}
		\mbox{Tr}(V(\wt E)) = \sum_{\ell=1}^{M-1} W(\wt \balpha_\ell, \wt \balpha_{\ell+1}) +W(\wt \balpha_{M},\delta_{L,1}\alpha_{N},0) +\underbrace{2\Re( \alpha_1^2+2\bar{\alpha}_2\rho_1^2)}_{= W(0,-1, \alpha_1,\alpha_2)}\,,
	\end{equation}
	where in this case $$W(\wt \balpha_\ell, \wt \balpha_{\ell+1}) =2 \Re\sum_{s=0}^1(\alpha_{2\ell-1+s}\bar{\alpha}_{2\ell+s})^2-4 \Re\sum_{s=0}^1\alpha_{2\ell-1+s}\bar{\alpha}_{2\ell+1+s}\rho_{2\ell+s}^2$$ and $\delta_{L,1}$ is equal to zero for $L\neq 1$.\\
	\noindent Using \eqref{eq:potentialE} the partition function can be written in the form
	\begin{equation}
		\label{eq:partition_2}
		\begin{split}
			Z^{HT}_N(V,\beta) &= \int\limits_{\D^{N-1}\times S^1} \frac{\di\alpha_N}{i\alpha_N}\prod_{j=1}^{N-1} \di^2 \alpha_j \left(1-\vert \alpha_j\vert ^2\right)^{\beta\left(1-\frac{j}{N}\right)-1}\\
			&\times\exp\left(- \sum_{\ell=1}^{M-1} W(\wt \balpha_\ell, \wt \balpha_{\ell+1}) - 
			W(\wt \balpha_{M}, \alpha_{KM+1}, \ldots,\alpha_N,\overbrace{0,\dots,0}^{K-L})\right) \\ &\times \exp\left(- W(\overbrace{0,\ldots,0}^{K-1},-1, \wt \balpha_1)\right)\,.
		\end{split}
	\end{equation}
	We want to apply the same technique as in the previous case, but we have to pay attention to one important detail: in this situation, the eigenvalues and the  eigenfunctions of the transfer operators will be dependent on the block number. Indeed, in this case, the  exponents of $(1-\vert \alpha_j\vert ^2)$ are not identical, but they depend on the index $j$  as in  \eqref{eq:partition_2}. 
	
	\noindent For this reason, we define 
	$$
	\bzeta^{(1)} = \beta\overbrace{\left(1 - \frac{1 }{N},1 - \frac{  2 }{N}\, \ldots, 1 - \frac{K }{N} ,1 - \frac{1 }{N},1 - \frac{  2 }{N}, \ldots, 1 - \frac{K }{N}  \right)}^{2K}\,,
	$$
	and 
	$$
	\bzeta^{(j)}=\bzeta^{(1)}-\beta\frac{j-1 }{N}\boldsymbol{K},\quad j=1,\dots, M-1\,,
	$$
	where  the vector $\boldsymbol{K}$ has entries  $\boldsymbol{K}_j=K$ for $j=1,\dots,2K$.
	For $K$ integer and $K<N$ we introduce the multiplication operator $\cM_K:L^2(\D^K)\to L^2(\D^K)$  defined as 
	
	$$(\cM_Kf)(\balpha)=\prod_{j=1}^{K} \left(1-\vert \alpha_j\vert ^2\right)^{-\frac{K\beta}{2N}} f(\balpha).
	$$
	\begin{remark}
		We notice that, for  $\beta \in\R$, $K\in \N$ and $N\in \N$ big enough, the function $\prod_{j=1}^{K} \left(1-\vert \alpha_j\vert ^2\right)^{-\frac{K\beta}{2N}}\in L^2(\D^K)$.  Since we are considering the limit $N\to\infty$,  and $\beta, K$ independent from $N$, we always assume that this condition holds.
	\end{remark}
	We observe that $\cM_{-K}=(\cM_K)^{-1}$ and the operators $\cT_{\bzeta^{(j)}}:L^2(\D^K)\to L^2(\D^K)$  defined in \eqref{eq:transfer} satisfy the relation
	\begin{equation}
		\label{eq:operator_relation}
		\cT_{\bzeta^{(j+1)}}=\cM_K\cT_{\bzeta^{(j)}}\cM_{K},\quad j=1,\dots,M-1.
	\end{equation}
	We recall that the operators $\cT_{\bzeta^{(j)}}$ are compact, furthermore, we notice that $\cM_K\cT_{\bzeta^{(j)}}$ is also compact since it is Hilbert--Schmidt  \cite{KatoBook}. 
	
	\noindent  { Let us define the $K(M-1)$-dimensional  vector $\bzeta_M=(\bzeta^{(M-1)},\dots,  \bzeta^{(1)})$ }and the 
	operator $ \wt \cT_{M,\bzeta_M}:L^2(\D^K)\to L^2(\D^K)$  as
	\begin{equation}
		\label{eq:T_tilde}
		\wt \cT_{M,\bzeta_M} = 	\cM_K\cT_{\bzeta^{(M-1)}}\cM_K \cT_{\bzeta^{(M-2)}}\cM_K\cdots \cM_K\cT_{\bzeta^{(1)}}\, ,
	\end{equation}
	we notice that it is a compact operator, since all $\cM_K\cT_{\bzeta^{(j)}}$ are Hilbert--Schmidt.
	
	We will now prove the following technical result:
	\begin{proposition}
		\label{prop:limit}
		Let $ \wt \cT_{M,\bzeta_M} $ as in \eqref{eq:T_tilde} and $Z_N^{HT}$ as in \eqref{eq:partition_2} then:
		\begin{equation}
			\label{eq:limit_bound}
			\lim_{N\to\infty}\frac{1}{N}\ln\left( \frac{Z_N^{HT}}{\mbox{Tr}( \wt \cT_{M,\bzeta_M} )}\right) = 0\,,
		\end{equation}
		here by $\mbox{Tr}( \wt \cT_{M,\bzeta_M} )$ we indicate the standard trace on $L^2$.
	\end{proposition}
	
	\begin{proof}
		We will estimate both $Z_N^{HT}$, and $\mbox{Tr}(\wt\cT_{M,\bzeta_M})$ from above and below, then combining these estimates we will obtain \eqref{eq:limit_bound}. We start with $Z_N^{HT}$.
		
		\begin{equation}
			\begin{split}
				Z_N^{HT} &= \int\limits_{\D^{N-1}\times S^1}\frac{\di\alpha_N}{i\alpha_N}\left[\prod_{j=1}^{N-1}  \di^2 \alpha_j\left(1-\vert \alpha_j\vert ^2\right)^{\beta\left(1-\frac{j}{N}\right)-1}\right]\\
				&\times\exp\left(- \sum_{\ell=1}^{M-1} W(\wt \balpha_\ell, \wt \balpha_{\ell+1}) - W(\wt \balpha_{M}, \alpha_{KM+1}, \ldots,\alpha_N,\overbrace{0,\dots,0}^{K-L})\right) \\ &\times \exp\left( - W(\overbrace{0,\ldots,0}^{K-1},-1, \wt \balpha_1)\right)\,.
			\end{split}
		\end{equation}
		We can bound the first and the last three terms in the above exponential  with two positive constants  $C(V,\beta)$ and $c(V,\beta)$, independent of $N$, such that
		\begin{equation}
			\begin{split}
				c(V,\beta) &\leq \exp\left( -  W(\wt \balpha_1, \wt \balpha_{2})-  W(\wt \balpha_{M-1}, \wt \balpha_{M})-W(\wt \balpha_{M}, \alpha_{KM+1}, \ldots,\alpha_N,\overbrace{0,\dots,0}^{K-L}) \right) \\ & \times \exp\left(- W(\overbrace{0,\ldots,0}^{K-1},-1, \wt \balpha_1)\right)  \leq C(V,\beta)\,    
			\end{split}
		\end{equation}
		where in the exponents  each $\alpha_j\in\overline{\mathbb{D}}$.
		From the previous inequalities, we deduce that  the integral
		\begin{equation}
			\int\limits_{\D^{N-1}\times S^1} \frac{\di\alpha_N}{i\alpha_N}
			\left[\prod_{j=1}^{N-1}  \di^2 \alpha_j\left(1-\vert \alpha_j\vert ^2\right)^{\beta\left(1-\frac{j}{N}\right)-1}\right]
			\exp\left(- \sum_{\ell=2}^{M-2} W(\wt \balpha_\ell, \wt \balpha_{\ell+1}) \right) 
		\end{equation}
		is bounded from above by $ Z_N^{HT} /c(V,\beta) $  and from below by $ Z_N^{HT} /C(V,\beta) $.
		We can explicitly integrate in  $\alpha_j$ for $j=1,\ldots,K$ and $j=(M-1)K + 1,\ldots,N$ using the formula
		\begin{equation}
			\label{eq:explicit_integral}
			\int_\D \left( 1-\vert z\vert ^2\right)^{t-1} \di^2z = \pi t^{-1}\, , 
		\end{equation}
		obtaining that there are two constants $C_1(V,\beta)$ and $c_1(V,\beta)$ depending on $V,\,\beta, K$ and $L$  but not on $N$, such that  
		
		\begin{equation}
			\label{eq:ub_ZHT}
			\begin{split}
				Z_N^{HT} &\leq C_1(V,\beta)N^{K+L-1} \int\limits_{\D^{(M-2)K}} \left[ \prod_{j=K+1}^{(M-1)K} \di^2\alpha_j\left(1-\vert \alpha_j\vert ^2\right)^{\beta\left(1-\frac{j}{N}\right)-1}\right]\\ & \times\exp\left(- \sum_{\ell=2}^{M-2} W(\wt \balpha_\ell, \wt \balpha_{\ell+1}) \right)\, ,
			\end{split}
		\end{equation}
		and 
		\begin{equation}
			\label{eq:lb_ZHT}
			\begin{split}
				Z_N^{HT} &\geq c_1(V,\beta)N^{K+L-1} \int\limits_{\D^{(M-2)K}} \left[ \prod_{j=K+1}^{(M-1)K} \di^2\alpha_j\left(1-\vert \alpha_j\vert ^2\right)^{\beta\left(1-\frac{j}{N}\right)-1}\right] \\ & \times    \exp\left(- \sum_{\ell=2}^{M-2} W(\wt \balpha_\ell, \wt \balpha_{\ell+1}) \right)\, .
			\end{split}
		\end{equation}
		
		\noindent We can proceed analogously to estimate the trace of $ \wt \cT_{M,\bzeta_M} $:
		\begin{equation}
			\begin{split}
				\mbox{Tr}( \wt \cT_{M,\bzeta_M} ) &= \int\limits_{\D^{(M-1)K}}\prod_{j=1}^K \left(1- \vert \alpha_j\vert ^2\right)^{\frac{\beta}{2}\left( 1- \frac{j}{N}\right) -\frac{1}{2}} \prod_{j=1}^K \left(1- \vert \alpha_j\vert ^2\right)^{\frac{\beta}{2}\left( 1- \frac{(M-1)K+ j}{N}\right) -\frac{1}{2}}\\ & \times \prod_{j=K+1}^{(M-1)K}\left( 1 - \vert \alpha_j\vert ^2\right)^{\beta\left(1-\frac{j}{N}\right)-1}\\
				&\times \exp\left(- \sum_{j=1}^{M-2}W(\wt \balpha_j, \wt \balpha_{j+1}) - W(\wt \balpha_{M-1}, \wt \balpha_1)\right) \times\prod_{j=1}^{(M-1)K}\di^2 \alpha_j\,.
			\end{split}
		\end{equation}
		As before, we notice that there exist two  positive  constants $\wt C(V,\beta)$, and $\wt c(V,\beta)$, independent of $N$, such that
		\begin{equation}
			\wt c(V,\beta)<\exp\left(-W(\wt \balpha_1, \wt \balpha_{2}) - W(\wt \balpha_{M-1}, \wt \balpha_1) \right) <\wt C(V,\beta)\,
		\end{equation}
		when $\balpha_1,\balpha_2,\balpha_{M-1}\in\overline{\mathbb{D}^K}$.
		From these inequalities, we deduce that  the integral
		\begin{equation}
			\begin{split}
				&\int\limits_{\D^{(M-1)K}}\prod_{j=1}^K \left(1- \vert \alpha_j\vert ^2\right)^{\frac{\beta}{2}\left( 1- \frac{j}{N}\right) -\frac{1}{2}} \prod_{j=1}^K \left(1- \vert \alpha_j\vert ^2\right)^{\frac{\beta}{2}\left( 1- \frac{(M-1)K+ j}{N}\right) -\frac{1}{2}}\\ & \times \prod_{j=K+1}^{(M-1)K}\left( 1 - \vert \alpha_j\vert ^2\right)^{\beta\left(1-\frac{j}{N}\right)-1} \exp\left(- \sum_{j=2}^{M-2}W(\wt \balpha_j, \wt \balpha_{j+1}) \right) \prod_{j=1}^{(M-1)K}\di^2 \alpha_j 
			\end{split}
		\end{equation}
		is bounded from above by $ \mbox{Tr}( \wt \cT_{M,\bzeta_M} )/ \wt c(V,\beta)$  and from below by $ \mbox{Tr}( \wt \cT_{M,\bzeta_M} )/ \wt C(V,\beta)$. 
		Using \eqref{eq:explicit_integral} we can now explicitly integrate in $\alpha_j$ for $j=1,\ldots,K$ the above integral obtaining the following inequalities
		\begin{equation}
			\label{eq:ub_trace}
			\begin{split}
				\mbox{Tr}( \wt \cT_{M,\bzeta_M} ) & \leq \wt C_1(V,\beta)\int\limits_{\D^{(M-2)K}}\left[ \prod_{j=K+1}^{(M-1)K} \di^2\alpha_j\left(1-\vert \alpha_j\vert ^2\right)^{\beta\left(1-\frac{j}{N}\right)-1}\right] \\ &\times \exp\left(- \sum_{\ell=2}^{M-2} W(\wt \balpha_\ell, \wt \balpha_{\ell+1}) \right)\,,
			\end{split}
		\end{equation}
		
		\begin{equation}
			\label{eq:lb_trace}
			\begin{split}
				\mbox{Tr}( \wt \cT_{M,\bzeta_M} ) &\geq \wt c_1(V,\beta)\int\limits_{\D^{(M-2)K}}\left[ \prod_{j=K+1}^{(M-1)K} \di^2\alpha_j\left(1-\vert \alpha_j\vert ^2\right)^{\beta\left(1-\frac{j}{N}\right)-1}\right] \\ & 
				\times \exp\left(- \sum_{\ell=2}^{M-2} W(\wt \balpha_\ell, \wt \balpha_{\ell+1}) \right)\, ,
			\end{split}
		\end{equation}
		where $\wt C_1(V,\beta)$, and  $ \wt c_1(V,\beta)$ are  positive constants depending on $V,\,\beta,K$ and $L$  but not on $N$.
		Combining \eqref{eq:ub_ZHT}-\eqref{eq:lb_ZHT}-\eqref{eq:ub_trace}-\eqref{eq:lb_trace} we deduce \eqref{eq:limit_bound}.
	\end{proof}
	
	Applying the previous proposition, we can express the Free energy of the Circular $\beta$  ensemble in the high-temperature regime in terms of $\mbox{Tr}(\wt\cT_M)$:
	\begin{equation}
		\label{eq:free_energy_trace_00}
		\begin{split}
			F_{HT}(V,\beta) &= - \lim_{N\to \infty} \frac{1}{N} \ln\left(Z_N^{HT} \right) \\
			&= - \lim_{N\to \infty} \frac{1}{N}\left( \ln\left(\frac{Z_N^{HT}}{\mbox{Tr}( \wt \cT_{M,\bzeta_M})} \right) + \ln(\mbox{Tr}( \wt \cT_{M,\bzeta_M}))\right) \\
			& = - \lim_{N\to\infty}\frac{\ln(\mbox{Tr}( \wt \cT_{M,\bzeta_M} ))}{N}\, ,
		\end{split}
	\end{equation}
	where in the last equality we used Proposition \ref{prop:limit}.
	
	As a final step,  we have to understand the behaviour of $\mbox{Tr}( \wt \cT_{M,\bzeta_M} )$,   and for this purpose  we need to carefully analyse the compact operators $\cT_{\bzeta^{(j)}}$.
	
	Let  $\{\psi_n(\bz,V,\bzeta^{(j)})\}_{n\geq 1}$  be the eigenfunctions of $\cT_{\bzeta^{(j)}}$ with corresponding eigenvalues $\{\lambda_{n}(V,\bzeta^{(j)})\}_{n\geq 1}$ and { $\vert \lambda_{1}(V,\bzeta^{(j)})\vert \geq \vert \lambda_{2}(V,\bzeta^{(j)})\vert \geq \dots$}. From a generalized version of Jentzsch's Theorem (see \cite[Theorem 137.4]{ZaanenBook}), we deduce that $\vert \lambda_n(V,\bzeta^{(j)})\vert  < \lambda_1(V,\bzeta^{(j)})$ for all $n\geq2$. 
	
	We are now in  the  position to prove the following proposition.
	
	\begin{proposition}
		\label{prop:stima_brutta}
		Let $\{\psi_n(\bz,V,\bzeta^{(j)})\}_{n=1}^\infty$ be the eigenfunctions of  the operator  $\cT_{\bzeta^{(j)}}$   in \eqref{eq:transfer} with corresponding eigenvalues $\{\lambda_n(V,\bzeta^{(j)})\}_{n=1}^\infty$. Consider the operator $ \wt \cT_{M,\bzeta_M} $ in \eqref{eq:T_tilde}, then there are constants $\td,a_j,c_j$, $j=1,\ldots,M-1${ uniformly bounded in $N$, and so in $M$, }such that :

		\begin{equation}
			\label{eq:first_goal}
			\left(\psi_1(\bz,V,\bzeta^{(1)}), \wt \cT_{M,\bzeta_M}  \psi_1(\bz,V,\bzeta^{(1)}) \right)= \prod_{j=1}^{M-1} \lambda_1(V,\bzeta^{(j)})\left( 1 + \frac{a_j}{N} + O\left(\frac{1}{N^2}\right)\right)\, ,
		\end{equation}
		\begin{equation}
			\label{eq:second_goal}
			\left\vert  \sum_{\ell\geq 2} \left( \psi_{\ell}(\bz,V,\bzeta^{(1)}), \wt \cT_{M,\bzeta_M}  \psi_{\ell}(\bz,V,\bzeta^{(1)}) \right)\right\vert  \leq\td
			\prod_{j=3}^{M-1}\lambda_1(V,\bzeta^{(j)})\left( 1 + \frac{c_j}{N} + O\left(\frac{1}{N^2}\right)\right)
		\end{equation}
	\end{proposition}
	\begin{proof}
		To simplify the notation, we will drop the $V$ dependence of the eigenvalues $\lambda_n(V,\bzeta^{(j)})$, and of the eigenfunctions $\psi_n(\bz,V,\bzeta^{(j)})$. 
		
		We will prove \eqref{eq:first_goal} by induction on $M$. 
		For $M=2$, we have that $ \wt \cT_{M,\bzeta_M}  = \cM_K\cT_{\bzeta^{(1)}}$, so we have to compute:
		
		\begin{equation}
			\begin{split}
				& \bigg(\psi_1(\bz,\bzeta^{(1)}), \cM_K \cT_{\bzeta^{(1)}}  \psi_1(\bz,\bzeta^{(1)}) \bigg)= \lambda_1(\bzeta^{(1)})\left(\psi_1(\bz,\bzeta^{(1)}), \cM_K \psi_1(\bz,\bzeta^{(1)}) \right)\\
				&\qquad=\lambda_1(\bzeta^{(1)})\left(\psi_1(\bz,\bzeta^{(1)}),\prod_{j=1}^K (1-\vert z_j\vert ^2)^{-\frac{K\beta}{2N}}\psi_1(\bz,\bzeta^{(1)}) \right)\\
				&\qquad=\lambda_1(\bzeta^{(1)})\left(\psi_1(\bz,\bzeta^{(1)}),\left( 1+ \frac{\wt a_1(\bz)}{N} + O\left( \frac{1}{N^2}\right)\right) \psi_1(\bz,\bzeta^{(1)}) \right)\\
				&\qquad=\lambda_1(\bzeta^{(1)})\left(1 + \frac{a_1}{N} + O\left( \frac{1}{N^2}\right)\right),
			\end{split}
		\end{equation}
		where the function $\wt a_1(\bz)$ is the first term of the expansion of 
		$\prod_{j=1}^K (1-\vert z_j\vert ^2)^{-\frac{K\beta}{2N}}$  in powers of $1/N$ and  the constant 
		$a_1=\left(\psi_1(\bz,\bzeta^{(1)}),\wt a_1(\bz)\psi_1(\bz,\bzeta^{(1)}) \right)$ is { uniformly bounded in} $N$.
		So the first inductive step is proved.
		
		\noindent For general $M$, 
		we  define the vector $\bzeta_{M-1}=(\bzeta^{(M-1)},\dots,  \bzeta^{(2)})$   so that 
		$$ \wt \cT_{M,\bzeta_M}= \wt \cT_{M-1,\bzeta_{M-1}}\cM_K\cT_{\bzeta^{(1)}}. $$
		Using the above relation  we obtain
		\begin{equation}
			\label{eq:induction_01}
			\begin{split}
				& \left(\psi_1(\bz,\bzeta^{(1)}), \wt \cT_{M,\bzeta_M}  \psi_1(\bz,\bzeta^{(1)}) \right) =\left(\psi_1(\bz,\bzeta^{(1)}),\wt \cT_{M-1,\bzeta_{M-1}}\cM_K\cT_{\bzeta^{(1)}} \psi_1(\bz,\bzeta^{(1)}) \right)\\ &= \lambda_1(\bzeta^{(1)}) \left(\psi_1(\bz,\bzeta^{(1)}),\wt \cT_{M-1,\bzeta_{M-1}} \cM_K \psi_1(\bz,\bzeta^{(1)}) \right)\,.
			\end{split}
		\end{equation}
		Thanks to  \cite[Chapter VII, Theorem 1.8]{KatoBook}, we know that the eigenfunctions $ \psi_1(\bz,\bzeta^{(j)})$ and the eigenvalues $\lambda_1(\bzeta^{(j)})$ are analytic functions of the parameter $\bzeta^{(j)}$, so, for $N$ big enough, there exists a function $\xi_{1}(\bz)\in L^2(\D^K)$ independent of $N$ such that:
		\begin{equation}
			\label{eq:def_xij}
			\psi_1(\bz,\bzeta^{(1)}) = \psi_1(\bz,\bzeta^{(2)})\left( 1 + \frac{\xi_1(\bz)}{N} + O\left( \frac{1}{N^2}\right)\right)\,
		\end{equation}
		and a constant $c_j$ such that
		\begin{equation}
			\label{eq:def_cj}
			\lambda_1(\bzeta^{(j+1)})=\lambda_1(\bzeta^{(j)})\left( 1 + \frac{c_j}{N} + O\left(\frac{1}{N^2}\right)\right).
		\end{equation}
		Using  \eqref{eq:def_xij}   and the expansion of the function defining the operator $\cM_K$ we   can expand \eqref{eq:induction_01} as:
		
		\begin{equation}
			\label{eq:induction_03}
			\begin{split}
				&\left(\psi_1(\bz,\bzeta^{(1)}),\wt \cT_{M-1,\bzeta_{M-1}}\cM_K\cT_{\bzeta^{(1)}}
				\psi_1(\bz,\bzeta^{(1)}) \right) \\ &=  \lambda_1(\bzeta^{(1)}) \left(\psi_1(\bz,\bzeta^{(2)}),\cT_{M-1,\bzeta_{M-1}}\psi_1(\bz,\bzeta^{(2)}) \right)  \\ &
				+ \frac{\lambda_1(\bzeta^{(1)})}{N}
				\Bigg(\psi_1(\bz,\bzeta^{(2)}), \wt \cT_{M-1,\bzeta_{M-1}}\psi_1(\bz,\bzeta^{(2)})\left( \wt a_1(\bz) + \xi_1(\bz)+ O\left( \frac{1}{N}\right)\right) \Bigg) \\ &
				+ \frac{\lambda_1(\bzeta^{(1)})}{N}  \Bigg(\psi_1(\bz,\bzeta^{(2)})\left(\xi_1(\bz)+ O\left( \frac{1}{N}\right)\right), \wt \cT_{M-1,\bzeta_{M-1}} \psi_1(\bz,\bzeta^{(2)})\Bigg)\,.
			\end{split}
		\end{equation}
		
		To bound the  last two terms in the above relation, we  use \eqref{eq:operator_relation}  and \eqref{eq:def_cj} so  that
		
		\begin{equation}
			\label{eq:first_ineq}
			\left\vert   \left\vert \cM_K \cT_{\bzeta^{(j)}} \right\vert \right\vert = \left\vert  \left\vert \cT_{\bzeta^{(j+1)}}\cM_K^{-1} \right\vert \right\vert  \leq \lambda_1(\bzeta^{(j+1)})=\lambda_1(\bzeta^{(j)})\left( 1 + \frac{c_j}{N} + O\left(\frac{1}{N^2}\right)\right) \; \,
		\end{equation}
		for $ j=2,\ldots,M-1$, here in the first inequality we use the fact that $\vert \vert \cM_K^{-1}\vert \vert  = 1$.
		
		Using  \eqref{eq:first_ineq}   we can bound the second term in the r.h.s of \eqref{eq:induction_03} by 
		\begin{equation}
			\label{eq:long}
			\begin{split}
				&\left\vert \left(\psi_1(\bz,\bzeta^{(2)}),\cM_K \cT_{\bzeta^{(M-1)}} \cM_K \ldots\cM_K\cT_{\bzeta^{(2)}}  \psi_1(\bz,\bzeta^{(2)}) \left( \wt a_1(\bz) + \xi_1(\bz) + O\left( \frac{1}{N}\right)\right) \right)\right\vert  \\ & \leq \left\vert\left\vert\psi_1(\bz,\bzeta^{(2)})\right\vert\right\vert_2\left\vert\left\vert\psi_1 (\bz,\bzeta^{(2)})\left( \wt a_1(\bz) + \xi_1(\bz) + O\left( \frac{1}{N}\right)\right)\right\vert\right\vert_2\\
				&\times \left \vert \left \vert\cM_K \cT_{\bzeta^{(M-1)}} \cM_K \ldots\cM_K\cT_{\bzeta^{(2)}} \right\vert\right \vert
				\\ &
				\leq \tc  \prod_{j=2}^{M-1}\left \vert \left \vert \cM_K \cT_{\bzeta^{(j)}}\right\vert \right\vert  \leq \tc \prod_{j=3}^{M}\lambda_1(\bzeta^{(j)})\leq  \tc \prod_{j=2}^{M-1}\lambda_1(\bzeta^{(j)})\left( 1 + \frac{c_j}{N} + O\left(\frac{1}{N^2}\right)\right)\,,\\
			\end{split}
		\end{equation}
		for some constant $\tc$ { uniformly bounded in $N$}. An analogous inequality can be obtained 
		for the second term   in \eqref{eq:induction_03}. Thus, applying the induction to the first term in  the r.h.s. of \eqref{eq:induction_03}, we deduce \eqref{eq:first_goal}.
		
		{We move to the proof of  \eqref{eq:second_goal}.  Applying \eqref{eq:first_goal}, we can estimate \eqref{eq:second_goal} as
			\begin{equation}
				\begin{split}
					\bigg\vert  \sum_{\ell\geq 2}& \bigg( \psi_{\ell}(\bz,\bzeta^{(1)}), \wt \cT_{M,\bzeta_M}  \psi_{\ell}(\bz,\bzeta^{(1)}) \bigg)\bigg\vert   \\ &= \left\vert  \sum_{\ell\geq 2} \left( \psi_{\ell}(\bz,\bzeta^{(1)}), \wt \cT_{M,\bzeta_M}  \psi_{\ell}(\bz,\bzeta^{(1)}) \right) \pm \left( \psi_{1}(\bz,\bzeta^{(1)}), \wt \cT_{M,\bzeta_M}  \psi_{1}(\bz,\bzeta^{(1)}) \right)\right\vert  \\ &
					\leq  \prod_{j=1}^{M-1} \lambda_1(\bzeta^{(j)})\left( 1 + \frac{a_j}{N} + O\left(\frac{1}{N^2}\right)\right) + \vert  \mbox{Tr}( \wt \cT_{M,\bzeta_M}) \vert\,.
				\end{split}
			\end{equation}
			Regarding the second term in the r.h.s of the above expression  we claim that  there exists a constant $c$ such that
			\begin{equation}
				\vert  \mbox{Tr}( \wt \cT_{M,\bzeta_M}) \vert \leq c
				\prod_{j=3}^{M-1}\lambda_1(\bzeta^{(j)})\left( 1 + \frac{c_j}{N} + O\left(\frac{1}{N^2}\right)\right)\,.
			\end{equation}
			To derive the above inequality first, we consider the operator $\cT_{2,1} = \cM_K\cT_{\bzeta^{(2)}}\cM_K\cT_{\bzeta^{(1)}}$, it is a compact operator and it is trace class since it is the composition of two different Hilbert--Schmidt operators. Let $\{\wt \lambda_n\}_{n\geq 1}$ be its eigenvalues numbered in such a way that $\{\vert \wt \lambda_n\vert \}_{n\geq 1}$ is a non  increasing  sequence and let $\{\varphi_n(\bz)\}_{n\geq 1}$  be  the corresponding eigenfunctions.   Then 
			
			\begin{equation}
				\begin{split}
					\vert  \mbox{Tr}( \wt \cT_{M,\bzeta_M}) \vert &= \left\vert \sum_{n\geq 1} \left( \varphi_n, \cM_K \cT_{\bzeta^{(M-1)}}\ldots\cM_K \cT_{\bzeta^{(3)}} \cT_{2,1} \varphi_n\right)\right\vert \\ 
					& \leq \sum_{n\geq 1} \vert \wt \lambda_{n} \vert\, \left\vert \left( \varphi_n, \cM_K \cT_{\bzeta^{(M-1)}}\ldots\cM_K \cT_{\bzeta^{(3)}} \varphi_n\right) \right\vert \\ &\leq \left \vert \left \vert\cM_K \cT_{\bzeta^{(M-1)}} \cM_K \ldots\cM_K\cT_{\bzeta^{(3)}} \right\vert\right \vert\sum_{n\geq 1} \vert \wt \lambda_{n} \vert  \,.
				\end{split}
			\end{equation}
			Since $\cT_{2,1}$ is trace class, it is a classical result that \cite{operatorbook}
			\begin{equation}
				\sum_{n\geq 1} \vert \wt \lambda_{n} \vert \leq \sum_{n\geq 1}  s_{n} := \vert\vert \cT_{2,1} \vert\vert_{\textrm{Trace}}\,,
			\end{equation}
			where $\{s_n\}_{n\geq 1}$  are  the singular values of the operator $\cT_{2,1}$. Furthermore, since $\cT_{2,1}$ is the composition of  two Hilbert--Schmidt operators, we have the following inequality
			\begin{equation}
				\vert\vert \cT_{2,1} \vert\vert_{\textrm{Trace}} \leq \vert \vert \cM_K\cT_{\bzeta^{(2)}}\vert \vert_{\textrm{HS}}\vert \vert\cM_K\cT_{\bzeta^{(1)}} \vert \vert_{\textrm{HS}} \leq\wt  c\,,
			\end{equation}
			where $\vert \vert \cdot \vert\vert_{\textrm{HS}}$ is the Hilbert--Schmidt norm and $\wt c$ is a positive constant uniformly bounded in $N$.
			Thus, applying the previous chain of inequalities and the same argument as in \eqref{eq:long}, we deduce that 
			\begin{equation}
				\vert  \mbox{Tr}( \wt \cT_{M,\bzeta_M}) \vert \leq c \prod_{j=3}^{M-1}\lambda_1(\bzeta^{(j)})\left( 1 + \frac{c_j}{N} + O\left(\frac{1}{N^2}\right)\right)\,,
			\end{equation}
			so we conclude our proof.
		}
		
	\end{proof}

	Applying Proposition \eqref{prop:stima_brutta} to \eqref{eq:free_energy_trace_00} we obtain that:
	
	\begin{equation}
		\label{eq:limite_FE}
		\begin{split}
			F_{HT}(V,\beta) &= - \lim_{N\to \infty} \frac{1}{N} \ln\left(  \mbox{Tr}( \wt \cT_{M,\bzeta_M})\right)=\\ &  - \lim_{N\to\infty}\frac{1}{N}\ln \left( \sum_{n\geq 1} \left(\psi_n(\bz, V,\bzeta^{(1)}), \wt \cT_{M,\bzeta_M} \psi_n(\bz, V,\bzeta^{(1)}) \right) \right) \\
			& = -  \lim_{N\to \infty}\Bigg[ \frac{1}{N} \ln\left(\prod_{j=1}^{M-1} \lambda_1(V,\bzeta^{(j)})\left( 1 + \frac{a_j}{N} + O\left(\frac{1}{N^2}\right)\right)\right) \\ &  +\frac{1}{N} \ln\left(1 + \frac{ \sum_{\ell\geq 2} \left( \psi_{\ell}(\bz,V,\bzeta^{(1)}), \wt \cT_{M,\bzeta_M}  \psi_{\ell}(\bz,V,\bzeta^{(1)}) \right)}{\prod_{j=1}^{M-1} \lambda_1(V,\bzeta^{(j)})\left( 1 + \frac{a_j}{N} + O\left(\frac{1}{N^2}\right)\right)} \right)\Bigg]\,,
		\end{split}
	\end{equation}
	applying the  inequality  \eqref{eq:second_goal} of  Proposition \eqref{prop:stima_brutta}, we deduce that the last term  in the above relation goes to zero as $N\to\infty$ and we obtain that
	\begin{equation}
		F_{HT}(V,\beta) = - \lim_{N\to \infty} \frac{1}{N} \ln\left(\prod_{j=1}^{M-1} \lambda_1(V,\bzeta^{(j)}) \right) \,.
	\end{equation}
	Since  {$\lambda_1(V,\bzeta^{(j)})$} is positive and an analytic  function of the parameter $\bzeta^{(j)}$, 
	we approximate the vector $\bzeta^{(j)}$ with the  vector $(1- \frac{jK}{N})\overbrace{(\beta,\beta,\dots \beta)}^{2K}=(1- \frac{jK}{N})\boldsymbol{\beta}$ and deduce that 
	$ \lambda_{1}(V,\bzeta^{(j)}) =  \lambda_{1}\left(V,\boldsymbol{\beta}\left( 1- \frac{jK}{N}\right)\right) + O(N^{-1})$. Therefore, we can rewrite  \eqref{eq:limite_FE} as 
	
	\begin{equation}
		\begin{split}
			F_{HT}(V,\beta)& = - \lim_{N\to \infty}\frac{1}{N}\sum_{j=1}^{M-1} \ln\left(   \lambda_{1}\left(V,\boldsymbol{\beta}\left( 1- \frac{jK}{N}\right)\right)\right) = \\ &- \frac{1}{K} \int_{0}^{1} \ln\left( \lambda_{1}\left(V,\boldsymbol{\beta}x\right)\right) \di x\, .
		\end{split}
	\end{equation}
	This, combined with \eqref{eq:FAL},   leads to \eqref{eq:rel_free}.  Moreover, as a consequence of the last relation, we deduce that $F_{HT}(V,\beta)$ is analytic in $\beta$ for $\beta>0$.

	We notice that the proof   of Proposition \ref{PROP:MOMENT_RELATION} is heavily based on the assumption that the potential $V$ that we are considering is of finite range, otherwise our approach would not work. 
	
	We now prove the moments relations \eqref{eq:moments_free_energy}.  For this purpose we have to prove the relations
	\begin{align}
		\label{eq:rel_simp_HT}
		\int_\T \cos( \theta m) \mu^\beta_{HT}(\di \theta) &= \partial_t  F_{HT}\left(V + \frac{t}{2} \Re(z^m),\beta\right)_{\vert_{t=0}}\, , \\
		\label{eq:rel_simp_AL}
		\int_\T \cos(\theta m) \mu^\beta_{AL}(\di \theta) &= \partial_t  F_{AL}\left(V + \frac{t}{2}\Re(z^m),\beta\right)_{\vert_{t=0}}\,. 
	\end{align}
	Analogous relation can be written for the imaginary part of the moments.
	We focus on \eqref{eq:rel_simp_HT}. From Remark \ref{rem:free_energy}, we know that $F_{HT}(V,\beta) = \cF^{(V;\beta)}(\mu^\beta_{HT}(\theta))$,   where   the functional $\cF^{(V,\beta)}$ is defined in \eqref{eq:functional} and  $\mu^\beta_{HT}(\theta)$ is the density of states of the Circular $\beta$  ensemble at high-temperature.
	We write the Euler-Lagrange equation for this functional, getting that $\mu^\beta_{HT}(\theta)$ satisfies:	
	\begin{equation}
		\label{eq:euler-lagrange}
		2V(\theta) -2\beta \int_\T \ln\left(\sin\left(\frac{\vert \theta-\gamma\vert }{2}\right)\right) \mu^\beta_{HT}(\gamma)\di \gamma + \ln(\mu^\beta_{HT}(\theta)) + C(V,\beta) = 0\, , 
	\end{equation}
	where $C(V,\beta)$ is a constant not depending on $\theta$. 
	
	Now let us consider the   functional  corresponding  to the  potential $\wt V(\theta) = V(\theta) + \frac{t}{2} \cos( m\theta)$:
	
	\begin{equation}
		\begin{split}
			\cF^{\left(V(\theta) + \frac{t}{2}\cos(m\theta), \beta\right)}(\mu) & = 2\int_{\T} V(\theta) \mu(\theta) \di \theta + t \int_{\T} \cos( m\theta) \mu(\theta) \di \theta \\
			& +  \int_{\T}\ln\left(\mu(\theta)\right) \mu(\theta)\di \theta \\ & - \beta \int\int_{\T\times \T} \ln\sin\left(\frac{\vert \theta-\gamma\vert }{2}\right)\mu(\theta)\mu(\gamma)\di \theta\di \gamma  +\ln(2\pi)\,.
		\end{split}
	\end{equation}
	Also this functional has a unique minimizer   that we denote by $\mu^{(t)}(\theta)$,  with $\mu^{(0)}(\theta) = \mu^\beta_{HT}(\theta)$. Evaluating the above   functional at $\mu^{(t)}(\theta)$, and computing its derivative at $t=0$,  we deduce the following relation:
	\begin{equation}
		\label{eq:functional_der}
		\begin{split}
			\partial_t 	\cF^{\left(V(\theta) + \frac{t}{2}\cos( m\theta), \beta\right)}&(\mu^{(t)})_{\vert_{t=0}} = 2\int_{\T} V(\theta) \partial_t \mu^{(t)}(\theta)_{\vert_{t=0}} \di \theta + \int_{\T} \cos( m\theta) \mu^\beta_{HT}(\theta) \di \theta  \\ &\quad  -  2\beta \int\int_{\T\times \T} \ln\sin\left(\frac{\vert \theta-\gamma\vert }{2}\right)\mu^\beta_{HT}(\gamma)\partial_t \mu^{(t)}(\theta)_{\vert_{t=0}}\di \theta\di \gamma  \\ &\quad + \int_{\T}\ln\left(\mu^\beta_{HT}(\theta)\right) \partial_t \mu^{(t)}(\theta)_{\vert_{t=0}}\di \theta\, .
		\end{split}
	\end{equation}
	Testing \eqref{eq:euler-lagrange} against $\partial_t \mu^{(t)}(\theta)_{\vert_{t=0}}$ we obtain 
	\begin{equation}
		\begin{split}
			& 2\int_{\T} V(\theta) \partial_t \mu^{(t)}(\theta)_{\vert_{t=0}} \di \theta   - 2\beta \int\int_{\T\times \T} \ln\sin\left(\frac{\vert \theta-\gamma\vert }{2}\right)\mu^\beta_{HT}(\gamma)\partial_t \mu^{(t)}(\theta)_{\vert_{t=0}}\di \theta\di \gamma  \\ & + \int_{\T}\ln\left(\mu^\beta_{HT}(\theta)\right) \partial_t \mu^{(t)}(\theta)_{\vert_{t=0}}\di \theta =0\,,
		\end{split}
	\end{equation}
	where we have used   $\int_{\T} \partial_t \mu^{(t)}(\theta) \di \theta = 0$. 
	Thus, we can simplify \eqref{eq:functional_der} as :
	\begin{equation}
		\partial_t 	\cF^{\left(V(\theta) + \frac{t}{2}\cos( m\theta), \beta\right)}(\mu)_{\vert_{t=0}} =   \int_{\T} \cos( m\theta)\mu^\beta_{HT}(\theta) \di \theta\, ,
	\end{equation}
	which is equivalent to \eqref{eq:rel_simp_HT}.

	To complete the proof of Proposition \ref{PROP:MOMENT_RELATION} we have to show that \eqref{eq:rel_simp_AL} holds.
	From the definition of mean density of states \eqref{eq:mean_density_general} we obtain that:
	
	\begin{equation}
		\label{eq:fatica_finale}
		\begin{split}
			\int_\T \cos( m\theta)\mu^\beta_{AL}(\di \theta)&= \lim_{N\to \infty} \frac{\meanval{\Re(\mbox{Tr}(\cE^m))}}{2N}
			\\& =  - \lim_{N\to \infty} \frac{\partial_t\left(Z_N^{(AL)}\left(V + \frac{t}{2}\Re(z^m),\beta\right)\right)_{\vert_{t=0}}}{NZ_N^{(AL)}(V ,\beta)}\, ,
		\end{split}
	\end{equation}
	where the expected value is taken with respect to the generalized Gibbs ensemble of the Ablowitz--Ladik lattice.  A similar equation holds for the imaginary part of the moment.
	
	Let's focus on the numerator, first we notice that we can assume that $\Re(z^m)$ and $V$ to have the same range $K$.  
	The more general case can be treated in the same way. Differentiating the partition function we obtain
	
	\begin{equation}
		\begin{split}
			\partial_t&\left(Z_N^{(AL)}\left(V + \frac{t}{2}\Re(z^m),\beta\right)\right)_{\vert_{t=0}} = \frac{1}{2}\int_{\D^N}\Re(\mathrm{Tr}(\cE^m))\prod_{j=1}^{N} \left(1-\vert \alpha_j\vert ^2\right)^{\beta-1}\\ &\qquad\times \exp\left(- \sum_{\ell=1}^{M} W(\wt \balpha_\ell, \wt \balpha_{\ell+1}){ -W_1(\alpha_{KM+1}, \ldots,\alpha_{KM+L}, \wt\balpha_1)}\right)  \di^2 \balpha\,.
		\end{split}
	\end{equation}
	Due to the structure of the measure and of the Lax matrix $\cE$, we deduce that there exist two   smooth  functions $g: \D^K\times \D^K \to \R$  and  $g_1 \,:\, \overline{\D}^{L}\times \overline{\D}^{K} \to \R $     such that 
	
	\begin{equation}
		\label{eq:primo_step}
		\begin{split}
			\partial_t&\left(Z_N^{(AL)}\left(V + \frac{t}{2}\Re(z^m),\beta\right)\right)_{\vert_{t=0}}   = \int_{\D^N} \di^2 \balpha \prod_{j=1}^{N} \left(1-\vert \alpha_j\vert ^2\right)^{\beta-1}\\ 
			&\qquad\times\left[\sum_{\ell=1}^{M} g(\wt \balpha_\ell, \wt \balpha_{\ell+1})+g_1(\alpha_{KM+1}, \ldots,\alpha_{KM+L}, \wt\balpha_1)\right]\\
			&\qquad\times \exp\left(- \sum_{\ell=1}^{M} W(\wt \balpha_\ell, \wt \balpha_{\ell+1}){ -W_1(\alpha_{KM+1}, \ldots,\alpha_{KM+L}, \wt\balpha_1)}\right)  \di^2 \balpha\,.
		\end{split}
	\end{equation} 
	Proceeding as in the proof of Proposition \ref{eq:primo_step}, defining the operator  $\cT_{\bbeta}^{(t)}$ as
	
	\begin{equation}
		\cT_{\bbeta}^{(t)} = \cT_{\bbeta}e^{tg(\balpha_{M-1},\balpha_{M})}\,,
	\end{equation}
	for $N$ big enough, \eqref{eq:primo_step} is asymptotic to
	
	\begin{equation}
		\partial_t\left(Z_N^{(AL)}\left(V + \frac{t}{2}\Re(z^m),\beta\right)\right)_{\vert_{t=0}}   \sim M \mathrm{Tr}\left( \partial_t(\cT_{\bbeta}^{(t)})_{\vert_{t=0}}\cT_{\bbeta}^{M-2}\right)\,.
	\end{equation}
	Following the same reasoning as in the previous proof, in view of the analyticity of $\lambda_1(V,\bbeta)$, we deduce that the previous equation is asymptotic to
	
	\begin{equation}
		\label{eq:derivata}
		\partial_t\left(Z_N^{(AL)}\left(V + \frac{t}{2}\Re(z^m),\beta\right)\right)_{\vert_{t=0}}   \sim M \lambda_1(V,\bbeta)^{M-2}\partial_t \lambda_1(V+t/2\Re(z^m),\bbeta)_{\vert_{t=0}}\,.
	\end{equation}
	Exploiting \eqref{eq:eigenvalues}-\eqref{eq:FAL} and \eqref{eq:derivata}, we can rewrite \eqref{eq:fatica_finale} as:
	
	\begin{equation}
		\begin{split}
			&  \lim_{N\to \infty} \frac{\partial_t\left(Z_N^{(AL)}\left(V + \frac{t}{2}\Re(z^m),\beta\right)\right)_{\vert_{t=0}}}{NZ_N^{(AL)}(V,\beta)} \\
			&\qquad=  \lim\limits_{N\to \infty} \frac{ M \lambda_1(V,\bbeta)^{M-2}\partial_t \lambda_1\left(V+\frac{t}{2}\Re(z^m),\bbeta\right)_{\vert_{t=0}}}{ N\sum_{\ell=1}^\infty \lambda^{M-1}_\ell(V,\beta)} \\
			&\qquad=  \frac{\partial_t\lambda_{1}\left(\left(V + \frac{t}{2}\Re(z^m),\beta\right)\right)_{\vert_{t=0}}}{K \lambda_{1}(V,\beta)}  \\ &\qquad = -\partial_t\left(F_{AL}\left(V + \frac{t}{2}\Re(z^m),\beta\right)\right)_{\vert_{t=0}}\, .
		\end{split}
	\end{equation}
	Thus, we have completed the proof of Proposition \ref{PROP:MOMENT_RELATION}.
	$\qedsymbol$
	
	\section{Proof of lemma \ref{LEM:UNIQUE}}
	\label{app:B}
	We   prove \eqref{eq:R_def} for $k=1$ and  the  cases  $k>1$ easily follow.  For convenience we consider the more general case of $(\lambda,\eta,\beta)\in\C\times\C\times\{z\in\C\vert \Re\,z>0\}$.
	
	Let us define $R_1^{(s)}= \begin{pmatrix}
		f_{s}& h_{s} \\
		p_{s} & q_{s}
	\end{pmatrix}$ where $s\geq  1$. If follows from \eqref{Rks} that  	\begin{equation}
		\begin{pmatrix}
			f_{s}& h_{s} \\
			p_{s} & q_{s}
		\end{pmatrix} = 	\begin{pmatrix}
			f_{s-1}& h_{s-1} \\
			p_{s-1} & q_{s-1}
		\end{pmatrix}
		\begin{pmatrix}
			1 + \frac{\lambda\beta \eta}{s(s+\beta)} & \frac{\eta^2}{s(s+\beta+1)}\\
			1 & 0
		\end{pmatrix}\, ,\quad s> 1.
	\end{equation}
	Note that  in the  case $\eta=0$ the lemma is trivially satisfied.
	We will show that  all   the sequences $\{f_s,h_s,p_s,q_s\}_{s\geq 1}$   converge as $s\to \infty$, moreover $h_s,q_s \xrightarrow{s\to\infty}0$.
	First of all, we notice that $h_{s} =\frac{\eta^2 f_{s-1}}{s(s+\beta+1)} $ and $q_s = \frac{\eta^2 p_{s-1}}{s(s+\beta+1)} $, thus the convergence to zero of these two sequences follows from the convergence of $p_s$ and $f_s$ as $s\to\infty$. Moreover, the terms of the sequences  $\{f_s,p_s\}_{s\geq 1}$ obey to the  3-terms recurrence:
	\begin{equation}
		\label{recurrence_f}
		f_s = \left(1 + \frac{\lambda\beta \eta}{s(s+\beta)}\right)f_{s-1} +  \frac{\eta^2 }{(s-1)(s+\beta)}f_{s-2}\, ,
	\end{equation} 
	and the same holds for $p_s$ in place of $f_s$. Thus, we have just to prove that the sequence $\{f_s\}_{s\geq 1}$ converges.
	We assume that  $(\lambda,\eta,\beta)\in\Omega$ where $\Omega\subset \C\times\C\times\{z\in\C\vert \Re\,z>0\}$  is a compact set. With this assumption we can give a  bound   to $\vert f_s\vert $ from  above as:
	\begin{equation}
		\vert f_s\vert  \leq \left(1 + \frac{2\eta^2 + \vert \lambda \beta\eta\vert }{s(s+\beta)}\right)\max\left(\vert f_{s-1}\vert ,\vert f_{s-2}\vert \right)\, .
	\end{equation}
	Inductively, we deduce that there exists a constant $C=C(\Omega)$ such that:
	\begin{equation}
		\label{eq:bound_abs}
		\vert f_s\vert  \leq C\prod_{\ell = 1}^s \left(1 + \frac{2\eta^2 + \vert \lambda \beta\eta\vert }{\ell(\ell+\beta)}\right) \leq  C\prod_{\ell = 1}^\infty \left(1 + \frac{2\eta^2 + \vert \lambda \beta\eta\vert }{\ell(\ell+\beta)}\right) \, .
	\end{equation}
	Furthermore, the infinite product on the right-hand side of \eqref{eq:bound_abs} is convergent by a classical result, see for example \cite[Chapter XIII, Lemma 1]{Lang1999}, this implies that the sequence $\{f_s\}_{s\geq 1}$ is uniformly bounded. Moreover, we have that:
	\begin{equation}
		\label{eq:cauchy}
		\vert f_{s+1} - f_{s}\vert  \leq \frac{\vert f_{s}\lambda\beta \eta\vert }{(s+1)(s+1+\beta)} +  \frac{\eta^2 \vert f_{s-1}\vert }{s(s-1+\beta)} \leq \wt C  \frac{\eta^2 + \vert \lambda \beta\eta\vert }{s(s-1+\beta)}\, ,
	\end{equation} 
	for some constant $\tilde{C}>0$  that depends on the compact set $\Omega$.
	This last equation implies that the sequence $\{f_s\}_{s\geq 1}$ is a Cauchy sequence, thus it is convergent. So we get the claim \eqref{Rks}. The claim \eqref{eq:R_recurrence}  easily follows from 
	\eqref{Rks}. 
	
	Regarding the differentiability in the parameters $\lambda$, $\eta$ and $\beta$,  it follows from \eqref{recurrence_f}  that $f_s=f_s(\lambda,\eta,\beta)$
	is analytic in $\Omega$. Since $f_s(\lambda,\eta,\beta)\to f(\lambda,\eta,\beta)$  as $s\to \infty$  uniformly, then by Weierstrasse convergence theorem, $f(\lambda,\eta,\beta)$
	is analytic in $\Omega$.
	
	$\qedsymbol$
	\vspace{20pt}
	
\end{appendices}

\noindent	
{\bf Acknowledgements}

\noindent
We thank  Thomas Kriecherbauer, Ken McLaughlin,  Gaultier Lambert and Herbert Spohn   for the many discussions and suggestions  during our time at MSRI. We thank  Rostyslav Kozhan for useful  comments on the manuscript.
This material is based upon work supported by the National Science Foundation under Grant No. DMS-1928930 while the author participated in a program hosted by the Mathematical Sciences Research Institute in Berkeley, California, during the Fall 2021 semester  	"Universality and Integrability in Random Matrix Theory and Interacting Particle Systems".

This project has received funding from the European Union's H2020 research and innovation programme under the Marie Sk\l odowska--Curie grant No. 778010 {\em  IPaDEGAN}.  TG acknowledges the support of GNFM-INDAM group and the  research project Mathematical Methods in NonLinear
Physics (MMNLP), Gruppo 4-Fisica Teorica of INFN.  G.M. is financed by the KAM grant number 2018.0344.

\vskip 1cm
\noindent
{\bf Added note}. Independently,  H. Spohn  \cite{Spohn_AL}   discovered the  connection between the  Ablowitz--Ladik lattice  and the circular $\beta$ ensemble at high-temperature. He also calculated  Generalized Gibbs Ensemble averaged field  and currents and the  associated  hydrodynamic equations.

\end{document}